\documentclass[peerreview]{IEEEtran}
\usepackage{amsmath,amssymb}
\usepackage{algorithmic}
\usepackage[ruled]{algorithm2e}
\usepackage{array}
\usepackage{stfloats}
\usepackage{graphicx}
\usepackage{booktabs}
\usepackage{multirow}
\usepackage{float}
\usepackage[dvipsnames]{xcolor}
\usepackage[colorlinks]{hyperref}
\hypersetup{citecolor=red}
\hypersetup{linkcolor=green}
\hypersetup{urlcolor=ProcessBlue}
\usepackage{url}
\newtheorem{proposition}{Proposition}
\newenvironment{proof}{{\it Proof.}}{\hfill$\square$\vspace{0.3cm}\\}
\newcommand{\prodscal}[2]{\left\langle#1,#2\right\rangle}
  \newcommand{\argmin}{\mathop{\vphantom{\min}\mathchoice
  {\vcenter{\hbox{ argmin}}}
  {\vcenter{\hbox{ argmin}}}{\mathrm{argmin}}{\mathrm{argmin}}}\displaylimits}
\usepackage[utf8]{inputenc} 
\newcommand{\matnorm}[1]{{\left\vert\kern-0.25ex\left\vert\kern-0.25ex\left\vert #1 
    \right\vert\kern-0.25ex\right\vert\kern-0.25ex\right\vert}}

\begin{document}
\title{Preconditioned P-ULA for Joint Deconvolution-Segmentation of Ultrasound Images~-- Extended Version}

\author{Marie-Caroline~Corbineau,
        Denis~Kouam\'e,
        Emilie~Chouzenoux,
        Jean-Yves~Tourneret,
        Jean-Christophe~Pesquet%    
\thanks{M.-C. Corbineau, E. Chouzenoux and J.-C. Pesquet are with the CVN, CentraleSup\'elec, INRIA Saclay, University Paris-Saclay, Gif-Sur-Yvette, France (e-mail: first-name.last-name@centralesupelec.fr)}%  
\thanks{D. Kouam\'e is with the IRIT, CNRS UMR 5505, University of Toulouse, Toulouse, France (e-mail: kouame@irit.fr).}% 
\thanks{J.-Y. Tourneret is with the IRIT, ENSEEIHT, T\'eSA, University of Toulouse, Toulouse, France (e-mail: Jean-Yves.Tourneret@enseeiht.fr).}%
}

\maketitle

\begin{abstract}
Joint deconvolution and segmentation of ultrasound images is a challenging problem in medical imaging. By adopting a hierarchical Bayesian model, we propose an accelerated Markov chain Monte Carlo scheme where the tissue reflectivity function is sampled thanks to a recently introduced proximal unadjusted Langevin algorithm. This new approach is combined with a forward-backward step and a preconditioning strategy to accelerate the convergence, and with a method based on the majorization-minimization principle to solve the inner nonconvex minimization problems. As demonstrated in numerical experiments conducted on both simulated and \textit{in vivo} ultrasound images, the proposed method provides high-quality restoration and segmentation results and is up to six times faster than an existing Hamiltonian Monte Carlo method.
\end{abstract}


\section{Introduction}
\label{sec:intro}
\IEEEPARstart{I}{n} medical ultrasound (US) imaging, useful information can be drawn from the statistics of the tissue reflectivity function (TRF) to perform segmentation~\cite{pereyra2012segmentation}, tissue characterization~\cite{bernard2006statistics}, or classification~\cite{alessandrini2011restoration}. Let $x\in\mathbb{R}^{n}$ and $y\in\mathbb{R}^n$ be the vectorized TRF and radio-frequency (RF) image, respectively. The following simplified model is used \cite{jensen1993deconvolution,ng2006modeling}
\begin{equation}
y=Hx+\omega,
\label{eq:model_y}
\end{equation}  
where $H\in\mathbb{R}^{n\times n}$ is a linear operator that models the convolution with the point spread function (PSF) of the probe, and $\omega\sim\mathcal{N}(0,\sigma^2\mathbb{I}_n)$, with $\mathcal{N}$ the normal distribution, and $\mathbb{I}_n$ the identity matrix in $\mathbb{R}^{n\times n}$. This paper assumes that the PSF is known, while $\sigma^2 >0$ is an unknown parameter to be estimated. 
The TRF is comprised of $K$ different tissues, which are identified by a hidden label field 
$z = (z_i)_{1 \leq i\leq n} \in\{1,\ldots,K\}^n$. 
For every $k\in\{1,\ldots,K\}$, the $k^{\mathrm{th}}$ region is modeled by a generalized Gaussian distribution ($\mathcal{GGD}$) \cite{alessandrini2011restoration,zhao2016joint}, which is parametrized by a shape parameter $\alpha_k\in [0,3]$, related to the scatterer concentration, and a scale parameter $\beta_k>0$, linked to the signal energy. Given $y$ and $H$, the aim is to estimate a deblurred image $x$ \cite{jensen1992deconvolution,michailovich2007blind}, as well as $\sigma^2$, $\alpha = (\alpha_k)_{1 \leq k \leq K}$, $\beta = (\beta_k)_{1 \leq k \leq K}$, and the label field $z$.
Due to the interdependence of these unknowns, it is beneficial to perform the deconvolution and segmentation tasks in a joint manner \cite{ayasso2010joint,pirayre2017hogmep}. This is achieved in \cite{zhao2016joint} by considering a hierarchical Bayesian model, which is used within a Markov chain Monte Carlo (MCMC) method~\cite{pereyra2016survey} to sample $x$, $\sigma^2$, $\alpha$, $\beta$, and $z$ according to the full conditional distributions. Despite promising results in image restoration and segmentation, the method in \cite{zhao2016joint} is of significant computational complexity, in particular due to the adjusted Hamiltonian Monte Carlo (HMC) method \cite{neal2011mcmc,robert2018accelerating} used to sample the TRF.
%sensitive to the stepsize and to the number of leapfrog steps, as well as being 
%
Recently, efficient and reliable stochastic sampling strategies have been devised \cite{durmus2018efficient,pereyra2016proximal,schreck2016shrinkage} using the proximity operator \cite{bauschke2017convex}, which is known as a useful tool for large-scale nonsmooth optimization \cite{combettes2011proximal}. 
In this work, we investigate an MCMC algorithm to perform the joint deconvolution and segmentation of US images, where the TRF is sampled with a scheme inspired from the proximal unadjusted Langevin algorithm (P-ULA)~\cite{pereyra2016proximal}. The latter generates samples according to an approximation of the target distribution without acceptance test, while being geometrically ergodic whereas classical unadjusted Langevin algorithms may have convergence issues.

\subsection{Main contributions}
Our contributions include \textit{i}) the proposition of an original accelerated preconditioned version of P-ULA (PP-ULA), which relies on the use of a variable metric forward-backward strategy~\cite{stuart2004conditional,chouzenoux2014variable}, \textit{ii}) an efficient solver based on the majorization-minimization (MM) principle to tackle the involved nonconvex priors, and \textit{iii})
%, as shown in numerical experiments, 
a new hybrid Gibbs sampler yielding a substantial reduction of the computational time needed to perform joint high-quality deconvolution and segmentation of both simulated and \textit{in vivo} US images.
\\
This article is organized as follows: Section~\ref{sec:model} describes the investigated Bayesian model and sampling strategy. Section~\ref{sec:p-ula} focuses on the proposed TRF sampling method. Numerical experiments are finally presented in Section~\ref{sec:expe}.

\section{Bayesian Model}
\label{sec:model}
\subsection{Priors}
Fig.~\ref{fig:bayesian_model} illustrates the hierarchical model used to perform a joint deconvolution-segmentation of ultrasound images.
The following likelihood function is derived from \eqref{eq:model_y}
\begin{equation}
p(y|x,\sigma^2)=\frac{1}{(2\pi\sigma^2)^{n/2}}\exp\left(-\frac{\|y-Hx\|^2}{2\sigma^2}\right).
\end{equation} 

\noindent The TRF is a mixture of $\mathcal{GGD}$s which, under the assumption that the pixel values are independent given $z$, leads to
% the following prior
\begin{small}
\begin{equation}
p(x|\alpha,\beta,z)=\prod_{i=1}^n\frac{1}{2\beta_{z_i}^{1/\alpha_{z_i}}\Gamma(1+1/\alpha_{z_i})}\exp\left(-\frac{|x_i|^{\alpha_{z_i}}}{\beta_{z_i}}\right).
\label{eq:ggd_distri}
\end{equation}
\end{small}

\begin{figure}
\centering
\includegraphics[width=0.5\textwidth]{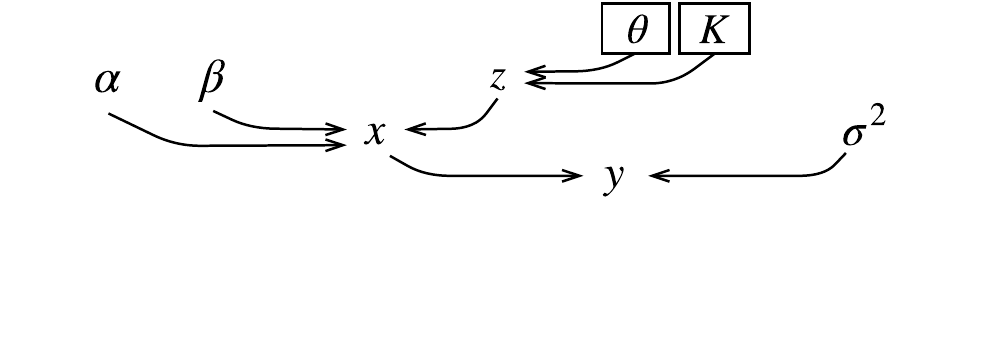}
\vspace*{-1.2cm}
\caption{Hierarchical Bayesian model. Parameters in boxes are fixed in advance.}
\label{fig:bayesian_model}
\end{figure}

\noindent Uninformative Jeffreys priors are assigned to the noise variance and scale parameters, while the shape parameters are assumed to be uniformly distributed between 0 and 3.
%\begin{equation}
%p(\sigma^2)=\frac{1}{\sigma^2}\imath_{[0,+\infty[}(\sigma^2),
%\end{equation}
%\begin{equation}
%p(\alpha)=\prod_{k=1}^K\frac{1}{3}\imath_{[0,3]}(\alpha_k),~~
%p(\beta)=\prod_{k=1}^K\frac{1}{\beta_k}\imath_{[0,+\infty[}(\beta_k),
%\end{equation}
%where $\imath$ is the characteristic function. 
The labels $z$ are modeled by a Potts Markov random field with prior
\begin{equation}
p(z)=\frac{1}{C(\theta)}\exp\left(\sum\nolimits_{i=1}^n\sum\nolimits_{j\in\mathcal{V}(i)}\theta\delta(z_i-z_j)\right),
\label{eq:distri_labels}
\end{equation}
with $\delta$ the Kronecker function, $C(\theta)>0$ a normalizing constant, $\theta>0$ a granularity coefficient, and $\mathcal{V}(i)$ the set of four closest neighbours of the $i^{\mathrm{th}}$ pixel. 
%A first order neighborhood of 4 nearest pixels is used in this work.
%
\subsection{Conditional distributions}
The different variables are sampled according to their conditional distributions, which are provided in this section.
The conditional distribution of the noise variance is derived from the Bayes theorem as follows
\begin{equation}
p(\sigma^2|y,x)\propto\mathcal{IG}\left(\frac{n}{2},\frac{\|y-Hx\|^2}{2}\right),
\label{eq:distri_sigma}
\end{equation}
%\begin{align}
%p(\sigma^2|y,x)&\propto p(y|x,\sigma^2)p(\sigma^2)\\
%&=\mathcal{IG}\left(\frac{n}{2},\frac{\|y-Hx\|^2}{2}\right),
%\label{eq:distri_sigma}
%\end{align}

\noindent where $\mathcal{IG}$ denotes the inverse gamma distribution. Assuming that the different regions have independent shape and scale parameters for every $k\in\{1,\ldots,K\}$, we obtain

\begin{align}
p(\alpha_k|x,\beta,z)&\propto \prod_{i\in\mathcal{I}_k}\frac{\mathbf{1}_{[0,3]}(\alpha_k)}{2\beta_{k}^{1/\alpha_{k}}\Gamma\left(1+1/\alpha_{k}\right)}\exp\left(-\frac{|x_i|^{\alpha_{k}}}{\beta_{k}}\right),
\label{eq:distri_alpha}\\
p(\beta_k|x,\alpha,z)&\propto \mathcal{IG}\left(\frac{n_k}{\alpha_k},\sum\nolimits_{i\in\mathcal{I}_k}|x_i|^{\alpha_k}\right),
\label{eq:distri_beta}
\end{align}
\noindent with $\mathcal{I}_k=\{i\in\{1,\ldots,n\} \mid z_i=k\}$, $n_k$ the number of elements in $\mathcal{I}_k$, and $\mathbf{1}_{[0,3]}$ the characteristic function of $[0,3]$. 
%Sampling the scale parameter directly from \eqref{eq:distri_beta} is straightforward, whereas 
Samples for $\alpha_k$ are drawn from \eqref{eq:distri_alpha} by using a Metropolis-Hastings (MH) random walk. 
% For every pixel 
For every pixel $i\in\{1,\ldots,n\}$ and every region $k\in\{1,\ldots,K\}$, the Bayes rule applied to the segmentation labels leads to
\begin{equation}
p(z_i=k|x,\alpha,\beta,z_{\mathcal{V}(i)})\propto \frac{\exp\left(\sum_{j\in\mathcal{V}(i)}\theta\delta(z_j-k)-\frac{|x_i|^{\alpha_k}}{\beta_k}\right)}{2\beta_k^{1/\alpha_k}\Gamma(1+1/\alpha_k)}\label{eq:label_notnormalized}
\end{equation}
%\begin{small}
%\begin{align}
%p(z_i=k|x,\alpha,\beta,z_{\mathcal{V}(i)})&\propto p(x_i|\alpha,\beta,z_i=k)p(z_i=k|z_{\mathcal{V}(i)})\\
%&\propto \frac{\exp\left(\sum_{j\in\mathcal{V}(i)}\delta(k-z_j)-\frac{|x_i|^{\alpha_k}}{\beta_k}\right)}{2\beta_k^{1/\alpha_k}\Gamma(1+1/\alpha_k)}\label{eq:label_notnormalized}
%\end{align}
%\end{small}
\noindent where $z_{\mathcal{V}(i)}$ denotes the label values in the neighborhood of $i$. As a consequence, the label $z_i$ is drawn from $\{1,\ldots,K\}$ using the above 
probabilities (suitably normalized).
%In the following, $\tilde{p}_{i,k}$ denotes the distribution~\eqref{eq:label_notnormalized} which has been normalized by $\sum_{l=1}^K p(z_i=l|x,\alpha,\beta,z_{\mathcal{V}(i)})$, so that $z_i$ is drawn from $\{1,\ldots,K\}$ with respective probabilities $(\tilde{p}_{i,k})_{1\le k \le K}$. 
%
%Let $\tilde{p}_{i,k}$ denote the normalized conditional distribution below
%\begin{equation}
%\tilde{p}_{i,k}=\frac{p(z_i=k|x,\alpha,\beta,z_{\mathcal{V}(i)})}{\sum_{l=1}^K p(z_i=l|x,\alpha,\beta,z_{\mathcal{V}(i)})}.\label{eq:distri_z}
%\end{equation}
% normalized by $\sum_{l=1}^K p(z_i=l|x,\alpha,\beta,z_{\mathcal{V}(i)})$.

\section{Preconditioned P-ULA}
\label{sec:p-ula}
\subsection{Notation}
Let $\mathcal{I}_{<1}=\{i\in\{1,\ldots,n\} \mid \alpha_{z_i}<1\}$ and $\mathcal{I}_{\geq 1}=\{1,\ldots,n\}\setminus \mathcal{I}_{<1}$. 
Let $\mathcal{S}_n$ denote the set of symmetric positive definite matrices in $\mathbb{R}^{n\times n}$, and let $\matnorm{\cdot}$ denote the spectral norm. For every $Q\in\mathcal{S}_n$, let $\|\cdot\|_Q=\prodscal{\cdot}{Q\cdot}^{1/2}$. For every function $f:\mathbb{R}^n\rightarrow\mathbb{R}\cup\{+\infty\}$, the proximity operator of $f$ at $x\in\mathbb{R}^n$ with respect to the norm induced by $Q^{-1}\in\mathcal{S}_n$ is defined as follows~\cite{bauschke2017convex},
\begin{equation}
\mathrm{prox}_{f}^Q(x)\in\mathrm{Argmin}_{u\in\mathbb{R}^n} \frac{1}{2}\|x-u\|^2_{Q^{-1}}+f(u). 
\label{eq:proxQ}
\end{equation}
If $Q$ is not specified, then $Q = \mathbb{I}_n$. If $\mathrm{prox}_f$ is simple to compute, then the solution to \eqref{eq:proxQ} for an arbitrary $Q \in \mathcal{S}_n$ can be obtained by using the dual forward-backward (DFB) algorithm~\cite{combettes2011proximity}, summarized in Algorithm~\ref{algo:dual_prox_metric}. If $f$ is proper, lower semicontinuous, and convex, then the sequence $(u^{(p)})_{p\in \mathbb{N}}$ generated by Algorithm~\ref{algo:dual_prox_metric} converges to $\mathrm{prox}_f^Q(x)$. 

\begin{algorithm}
\SetAlgoLined
Initialize dual variable $w^{(1)}\in\mathbb{R}^N$\;
Set $\rho=\matnorm{Q}^{-1}$, $\epsilon\in]0,\min\{1,\rho\}[$, $\eta\in[\epsilon,2\rho-\epsilon]$\;
\For{$p=1,...$}{
$u^{(p)}=x-Qw^{(p)} $\;
$w^{(p+1)}=w^{(p)}+\eta u^{(p)}-\eta\mathrm{prox}_{\eta^{-1} f}(\eta^{-1} w^{(p)}+u^{(p)})$
}
\caption{DFB algorithm to compute $\mathrm{prox}_f^Q(x)$ }\label{algo:dual_prox_metric}
\end{algorithm}
\vspace*{-0.3cm}
%
%=============================================
\subsection{Sampling the TRF}
The conditional distribution of the TRF is
\begin{equation}
\pi (x)= p(x|y,\sigma^2,\alpha,\beta,z)\propto \exp\left(-\frac{\|y-Hx\|^2}{2\sigma^2}-g(x)\right),
\label{eq:TRF_distri}
\end{equation}
\noindent where $(\forall x\in\mathbb{R}^n)$ $g(x)=\sum_{i=1}^n\beta_{z_i}^{-1}|x_i|^{\alpha_{z_i}}$. Let $\gamma>0$ and let $Q\in \mathcal{S}_n$ be a preconditioning matrix used to accelerate the sampler~\cite{marnissi2018majorize}. Following~\cite{pereyra2016proximal}, $\pi(x)$ is approximated by
\begin{equation}
\pi_\gamma(x) \propto \sup_{u\in\mathbb{R}^n}\pi(u)\exp\left(-\frac{\|u-x\|_{Q^{-1}}^2}{2\gamma}\right).
\label{eq:approx_target}
\end{equation}
As shown in Appendix, the Euler discretization of the Langevin diffusion equation~\cite{roberts2002langevin} applied to $\pi_{\gamma}$ with stepsize~$2\gamma$ and preconditioning matrix
$Q$ leads to
\begin{equation}
x^{(t+1)}= \mathrm{prox}_{\gamma g}^Q(\tilde{x}^t) +\sqrt{2\gamma}Q^{\frac{1}{2}}\omega^{(t+1)},
\label{eq:p-ula_it}
\end{equation}
where $\omega^{(t+1)}\sim\mathcal{N}(0,\mathbb{I}_n)$ and 
\begin{equation}
\tilde{x}^{(t)}=x^{(t)}-\frac{\gamma}{\sigma^2}QH^\top(Hx^{(t)}-y).
\label{eq:forward}
\end{equation}
Since the proposed sampling strategy is unadjusted, \eqref{eq:p-ula_it} is not followed by an acceptance test. The bias with respect to $\pi$ increases with $\gamma$, as the speed of convergence of the algorithm. A compromise must be found when setting $\gamma$. 

When $\mathcal{I}_{<1}$ is not empty, we use the MM principle~\cite{schifano2010majorization} to replace the nonconvex minimization problem involved in the computation of
$\mathrm{prox}_{\gamma g}^Q$ with a sequence of convex surrogate problems.
Let $\mathcal{J}\subset\mathcal{I}_{<1}$. We define $h_{\mathcal{J}}$ at every $(u,v)\in\mathbb{R}^n\times \mathbb{R}_{+*}^n$ by
\begin{equation*}
h_{\mathcal{J}}(u,v)=\sum\limits_{i\in\mathcal{I}_{\geq 1}} \frac{|u_i|^{\alpha_{z_i}}}{\beta_{z_i}}+\sum\limits_{j\in\mathcal{J}} \frac{(1-\alpha_{z_j})v_j^{\alpha_{z_j}}+\alpha_{z_j}v_j^{\alpha_{z_j}-1}|u_j|}{\beta_{z_j}}.
\end{equation*}
\noindent From concavity, we deduce that, for every $v\in\mathbb{R}_{+*}^n$ and $u\in\mathbb{R}^n$ such that $\mathcal{J}\subset  \{i\in\mathcal{I}_{<1}\mid |u_i|>0\}$, the following majoration property holds
\begin{equation*}
h_{\mathcal{J}}(u,v) \geq \sum\limits_{i\in\mathcal{I}_{\geq 1}\cup \mathcal{J}} \frac{|u_i|^{\alpha_{z_i}}}{\beta_{z_i}} = h_{\mathcal{J}}(u,(|u_i|)_{1\leq i\leq n}).
\end{equation*}
Since $h_{\mathcal{J}}(\cdot,v)$ is convex and separable, its proximity operator in the Euclidean metric is straightforward to compute. More precisely, 
%$\mathrm{prox}_{\eta^{-1}|\cdot|}$ amounts to soft-thresholding while 
for every $i\in\mathcal{I}_{\geq 1}$, $\eta>0$ and $s\in\mathbb{R}$, $\mathrm{prox}_{\eta^{-1}|\cdot|^{\alpha_{z_i}}}(s)$ has either a closed form \cite{Chaux2007} or can be found using a bisection search in $[0,|s|]$. Algorithm~\ref{algo:dual_prox_metric} can then be called, in order to compute the proximity operator of $h_{\mathcal{J}}(\cdot,v)$ in any metric $Q \in \mathcal{S}_n$. This leads to Algorithm~\ref{algo:MM} which generates a sequence $(u^{(q)})_{q\in\mathbb{N}}$ estimating $\mathrm{prox}_{\gamma g}^Q(\tilde{x}^{(t)})$.
\begin{algorithm}
\SetAlgoLined
Initialize $u^{(1)}\in\mathbb{R}^n$\;
\For{$q=1,...$}{
$\mathcal{J}^{(q)}=\{i\in\mathcal{I}_{<1} \mid |u_i^{(q)}|>0\}$\; 
$v^{(q)}= (|u_i^{(q)}|)_{1\leq i\leq n} $\;
$u^{(q+1)}= \mathrm{prox}^Q_{\gamma h_{\mathcal{J}^{(q)}}(\cdot,v^{(q)})}(\tilde{x}^{(t)})$ (using Alg.~\ref{algo:dual_prox_metric})
}
\caption{MM principle to compute $\mathrm{prox}_{\gamma g}^Q$.}
\label{algo:MM}
\end{algorithm}

The resulting Gibbs sampler is summarized in Algorithm~\ref{algo:gibbs}.
\begin{center}
\begin{algorithm}
\SetAlgoLined
\item[1]Sample the noise variance $\sigma^2$ according to \eqref{eq:distri_sigma}\;
\item[2]Sample the shape parameter $\alpha$ using MH with \eqref{eq:distri_alpha}\;
\item[3]Sample the scale parameter $\beta$ according to \eqref{eq:distri_beta}\;
\item[4]Sample the hidden label field $z$ using
\eqref{eq:label_notnormalized} 
%\eqref{eq:distri_z}
\;
\item[5]Sample the TRF $x$ using PP-ULA \eqref{eq:p-ula_it}-\eqref{eq:forward}.
\caption{Hybrid Gibbs sampler}
\label{algo:gibbs}
\end{algorithm}
\end{center}

\section{Numerical experiments}
\label{sec:expe}
\subsection{Experimental settings}
Six experiments are presented.
Simu1 and Simu2 refer to simulated images with two and three regions, respectively. Kidney denotes the tissue-mimicking phantom produced from $10^6$ scatterers uniformly distributed over a digital image of human kidney tissue provided with the Field II ultrasound simulator~\cite{jensen2004simulation}.
The amplitude of each scatterer is produced using a zero-mean Gaussian distribution whose variance is linked to the amplitude of the point on the digital image.
The PSF for the aforementioned simulations is obtained with Field II and corresponds to a 3.5~MHz linear probe. 
We also perform tests on three real ultrasound images. Thyroid denotes a real RF image of thyroidal flux obtained \textit{in vivo} with a 7.8~MHz probe. The unknown PSF is identified using the RF image of a wire cross-section which was acquired with the same probe. Since the diameter of the wire is of the order of a few $\mu $m, its cross-section can almost be viewed as a point. Thus, its RF image provides a good approximation of the PSF.
Finally, Bladder and KidneyReal refer to the RF images of a mouse bladder and mouse kidney, respectively. Both images were obtained \textit{in vivo} with a 20~MHz probe. The PSF for these two real images is estimated using the same method as for Thyroid.
The number of regions $K$ is set to 2 for Simu1 and KidneyReal, and it is set to 3 for Simu2, Kidney, Thyroid and Bladder.
The test settings and images can be found in Table~\ref{tab:test_settings} and Figs.~\ref{fig:TRF_simu} and \ref{fig:TRF_real} (first column), respectively. 

\begin{table}
\centering
\begin{tabular}{crcccccc}
\toprule
       && Simu1 & Simu2 & Kidney & Thyroid & Bladder & KidneyReal \\
\midrule
&Size   & 256$\times$256 & 256$\times$256 & 294$\times$354 & 870$\times$140 & 370$\times$256 & 350$\times$200\\
&Data type & Simulated & Simulated & Tissue-mimicking & Real \textit{in vivo} & Real \textit{in vivo} & Real \textit{in vivo}\\
\midrule
\multirow{3}{*}{Ground-truth}&TRF &  \checkmark &  \checkmark &  \checkmark & - & - & -\\
&$\mathcal{GGD}$ parameters &  \checkmark &  \checkmark & - & - & - & -\\
&Segmentation &  \checkmark &  \checkmark & - & - & - & -\\
\bottomrule \\
\end{tabular}
\caption{Test settings: size of test images, data type,  and availability of the ground-truth.}
\label{tab:test_settings}
\end{table}

The TRF is initialized using a pre-deconvolved image obtained with a Wiener filter, while the segmentation is initialized by applying a $7\times 7$ median filter and the Otsu method \cite{otsu1979threshold} to the B-mode of the initial TRF. Shape and scale parameters are randomly selected in $[0.5,1.5]$, and $[1,200]$, respectively. The granularity parameter $\theta$ for the Potts model \eqref{eq:distri_labels} is adjusted to ensure that the percentage of isolated points in the segmentation, obtained with a $3\times 3$ median filter, is close to 0.05, 0.1, 0.8, 0.08, 0.08 and 0.08 for Simu1, Simu2, Kidney, Thyroid, Bladder and KidneyReal, respectively.

\subsection{Comparisons and evaluation metrics}
All computational times are given for simulations run on Matlab 2018b on an Intel Xeon CPU E5-1650 3.20GHz. The code for the proposed method is available online\footnote{\url{https://github.com/mccorbineau/PP-ULA}}.
In addition to comparing Algorithm~\ref{algo:gibbs} with HMC~\cite{zhao2016joint}, the quality of the deconvolution is compared with the one obtained with a Wiener filter, where the noise level has been estimated as in~\cite{mallat1999wavelet}, and with the solution to the Lasso problem, where the regularization weight is set \textit{i}) manually when the ground-truth is not available, or \textit{ii}) using a golden-section search to maximize the peak signal-to-noise ratio (PSNR) defined as (with $x^{\rm tr}$ the true TRF and $x^{\rm es}$ the estimated one)
\begin{equation}
\mathrm{PSNR}=10\log_{10}(n\,\max\nolimits_i(x^{\rm tr}_i,x^{\rm es}_i)^2/\|x^{\rm tr}-x^{\rm es}\|^2).
\end{equation}
We also compare our results with the segmentation given by Otsu's method~\cite{otsu1979threshold} applied to the Wiener-deconvolved image, and with the SLaT method~\cite{cai2017three} 
%, adapted for grayscale images and 
applied to the Lasso-deconvolved image.
PP-ULA is used with $\gamma=0.09$ and $Q$ an approximation of the inverse of the Hessian of the differentiable term in \eqref{eq:TRF_distri} \cite{becker2012quasi}, $Q=\sigma^2(H^\top H +\lambda\mathbb{I}_n)^{-1}$, with $\lambda=0.1$ so that $Q$ is well-defined. 
We have also computed the structural similarity measure (SSIM)~\cite{wang2004image} of the restored TRF and the contrast-to-noise ratio (CNR) \cite{krishnan1997improved} between two windows from different regions of the B-mode TRF images.
%
%$\mathrm{CNR}=|\mu_1-\mu_2|/(\nu_1+\nu_2)^{1/2}$, with the means $(\mu_1,\mu_2)$ and variances $(\nu_1,\nu_2)$ in two windows from different regions of the B-mode image.
The segmentation is evaluated according to the percentage of correctly predicted labels,
%also called 
or overall accuracy (OA). The minimum mean square error (MMSE) estimators of all parameters in HMC and PP-ULA are computed after the burn-in regime. 
Moreover, to evaluate the mixing property of the Markov chain after convergence, we compute the mean square jump (MSJ) per second, which is the ratio of the MSJ to the time per iteration. The MSJ is obtained using $T$ samples of the TRF $(x^{t_0+1},\ldots,x^{t_0+T})$ generated after the burn-in period, 
i.e. $
\mathrm{MSJ}=\left(\frac{1}{T-1}\sum\nolimits_{t=1}^{T-1}\|x^{(t_0+t)}-x^{(t_0+t+1)}\|^2\right)^{1/2}.
$

%================================================================
\subsection{Results on simulated data}

The convergence speed of Algorithm~\ref{algo:gibbs} is empirically observed for Simu1 and Simu2, as illustrated in Fig.~\ref{fig:psnr_vs_time}, where we also display the results of the non-preconditioned P-ULA, for which $Q=\mathbb{I}_n$ and $\gamma=1.99\sigma^2/\matnorm{H}^2$. Comparing P-ULA and PP-ULA on these simulated data allows us to study the effect of adding a preconditioner in the proposed sampling scheme. As reported in Table~\ref{tab:res_simu1_simu2}, P-ULA needs more iterations and more time to converge than PP-ULA: the proposed method is 12.2 and 4.8 times faster than P-ULA on Simu1 and Simu2, respectively. In addition, from Table~\ref{tab:sig_shape_scale} and Fig.~\ref{fig:GGD_Simu1}, we deduce that P-ULA is more biased than PP-ULA, which samples correctly the target distributions. Finally, as one can see in Fig.~\ref{fig:psnr_vs_time} and Table~\ref{tab:psnr_ssim_cnr_simu1_simu2}, P-ULA leads to lower PSNR, SSIM and OA values than PP-ULA. These results clearly emphasize the benefits of preconditioning in this example. 

\begin{table}[!h]
\centering
\setlength\tabcolsep{3pt}
\begin{tabular}{lcrcccccccc}
\toprule
 && && \multicolumn{2}{c}{Iterations} && \multicolumn{2}{c}{Time} && \multicolumn{1}{c}{Mixing property}\\
 \cmidrule{5-6} \cmidrule{8-9} \cmidrule{11-11}
 && && Burn-in & Total && Duration & PP-ULA speed gain && MSJ (per s)\\
 \midrule
 \multirow{3}{*}{Simu1}&& P-ULA && 70000 & 140000 && 2~h~27~min & 12.2 && 665\\
 &&HMC   && 4000  & 8000   && 1~h~08~min & 5.7  && 173\\
 &&PP-ULA&& 4000  & 8000   && 12~min     & 1    && 970\\
 \midrule
  \multirow{3}{*}{Simu2}&& P-ULA && 70000 & 140000 && 3~h~06~min & 4.8 && 590\\
 &&HMC   && 10000  & 20000   && 4~h~14~min & 6.6  && 22\\
 &&PP-ULA&& 10000  & 20000   && 39~min     & 1    && 793\\
 \bottomrule\\
\end{tabular}
\caption{Number of iterations, computational time and MSJ per s for experiments Simu1 and Simu2.}\vspace*{-0.5cm}
\label{tab:res_simu1_simu2}
\end{table}

\begin{table}[!htb]
\centering
\setlength\tabcolsep{3pt}
\begin{tabular}{rcccccccccccccc}
\toprule
 && \multicolumn{5}{c}{Simu1} && \multicolumn{7}{c}{Simu2}\\
 \cmidrule{3-7} \cmidrule{9-15} 
 && $\sigma^2$ & $\alpha_1$ & $\beta_1$ & $\alpha_2$ & $\beta_2$ && 
 $\sigma^2$ & $\alpha_1$ & $\beta_1$ & $\alpha_2$ & $\beta_2$ & $\alpha_3$ & $\beta_3$\\
 \midrule
 True && 0.013 & 1.5 & 1.0 & 0.60 & 1.0    && 33 & 1.5 & 100  & 1.0 & 50  & 0.50 & 4.0\\
 P-ULA && 0.041 & 2.0 & 0.5 & 0.59 & 1.0   && 122 & 2.0 &  330 & 2.0 & 3186 & 0.48 & 3.4\\
 HMC && 0.013 & 1.8 & 1.2 & 0.61 & 1.0     && 34 & 1.4 & 66   & 1.1 & 111 & 0.54 & 5.2\\
 PP-ULA && 0.013& 1.4 & 0.9 & 0.62 & 1.1   && 35 & 2.3 & 2676 & 1.2 & 122 & 0.55 & 5.8 \\
 \bottomrule\\
\end{tabular}
\caption{MMSE Estimates of the noise variance and $\mathcal{GGD}$ parameters.}
\label{tab:sig_shape_scale}
\end{table}

\begin{figure}[!htb]
\centering
%\begin{tabular}{c}
\includegraphics[height= 0.2\textwidth, width=0.6\textwidth]{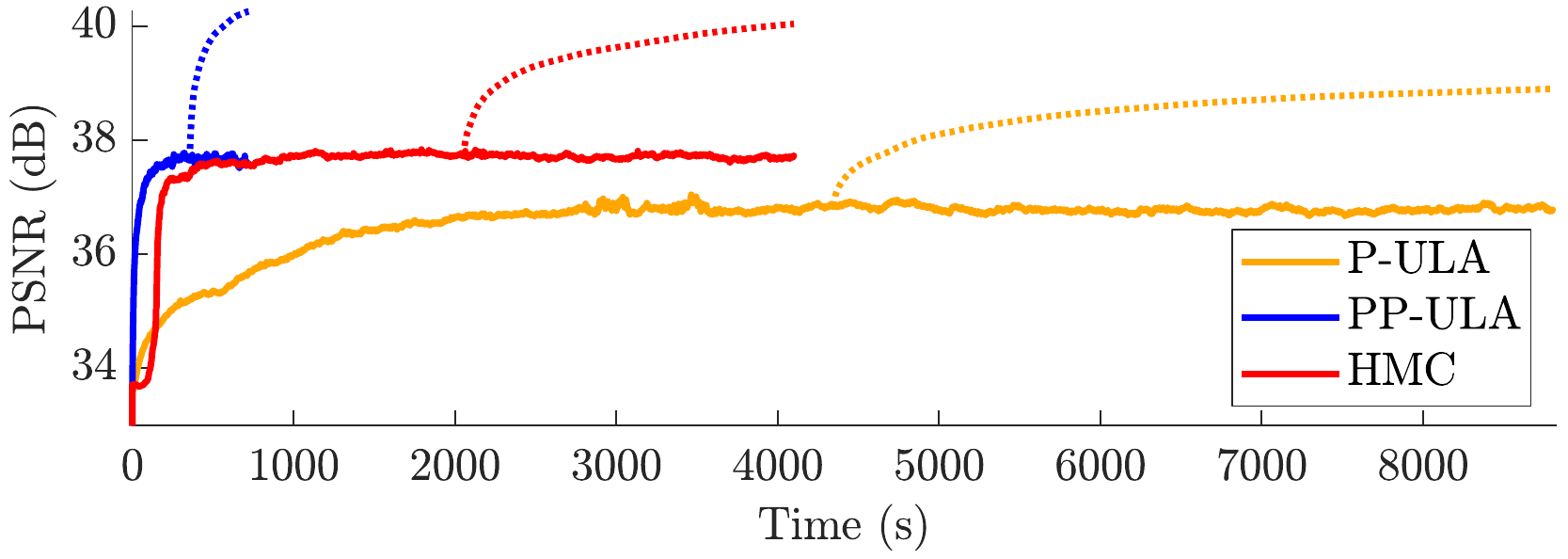}
%\\
%\includegraphics[height= 0.15\textwidth, width=0.45\textwidth]{pics/3regions/psnr_3regions.pdf}
%\end{tabular}
\vspace*{-0.3cm}
\caption{PSNR along time for Simu1. Dotted lines indicate the PSNR of the MMSE estimator of the TRF after the burn-in regime.}
\label{fig:psnr_vs_time}
\end{figure}

\begin{figure}[!htb]
\centering
\begin{tabular}{ccc}
\hspace*{-0.2cm}
\includegraphics[width=0.33\textwidth, height=0.2\textwidth]{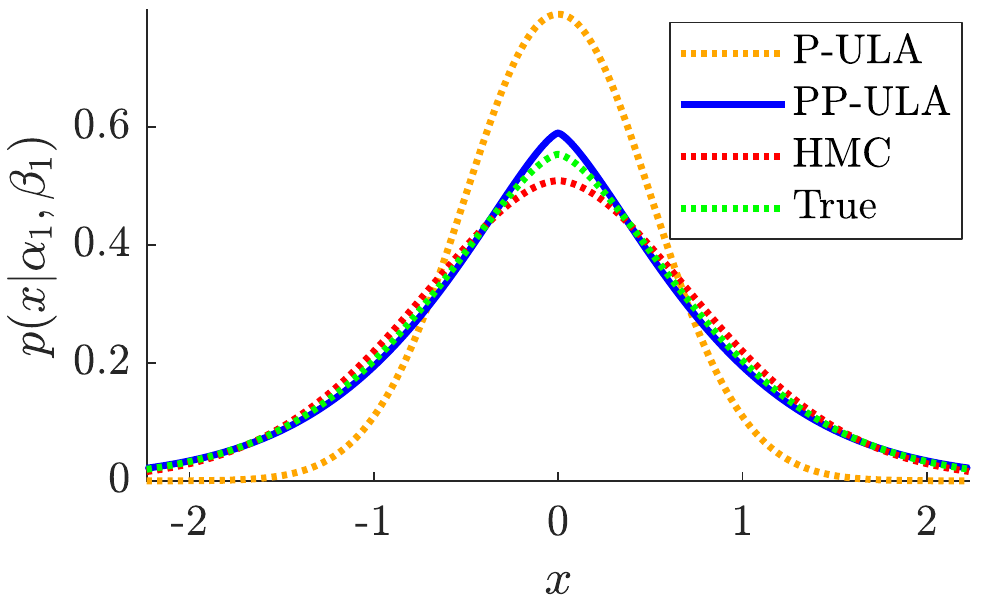}&
\includegraphics[width=0.33\textwidth, height=0.2\textwidth]{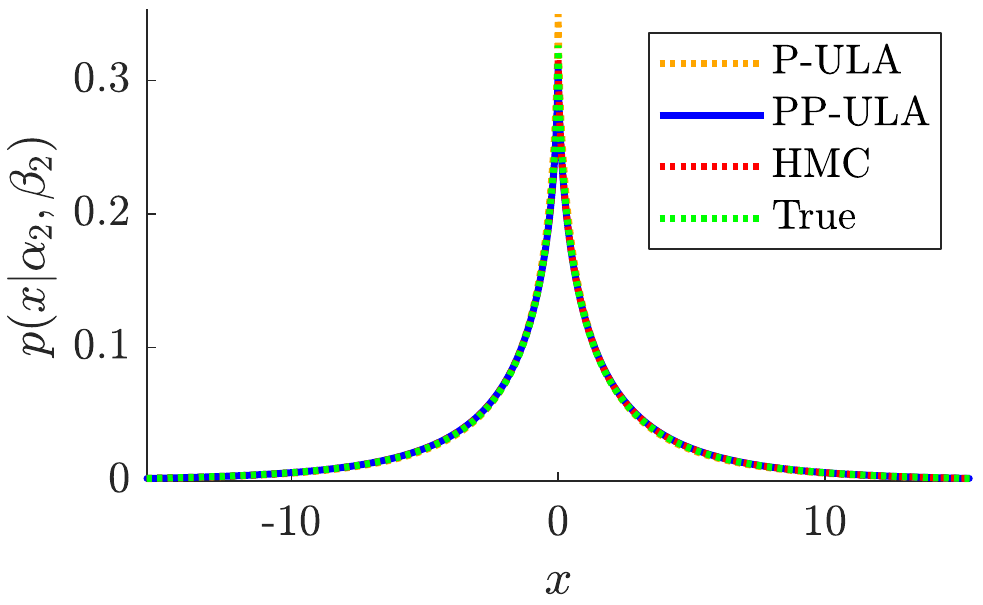}
%&\includegraphics[width=0.33\textwidth, height=0.2\textwidth]{pictures/Simu1/distri_lab2_Simu1_zoom.pdf}\\
\end{tabular}
\vspace*{-0.2cm}
\caption{Simu1, $\mathcal{GGD}$ distributions  \eqref{eq:ggd_distri} of regions 1 (left) and 2 (right).}
\label{fig:GGD_Simu1}
\end{figure}

\begin{table}[!htb]
\centering
\setlength\tabcolsep{3pt}
\begin{tabular}{rccccccccc}
\toprule
%%%% Simu1 and Simu2
 & \multicolumn{4}{c}{Simu1} & & \multicolumn{4}{c}{Simu2} \\
 \cmidrule{2-5}  \cmidrule{7-10} 
         & PSNR & SSIM & CNR & OA &&   PSNR & SSIM & CNR & OA \\
     \midrule    
 Wiener - Otsu & 37.1& 0.57 & 1.26 & 99.5 && 35.4 & 0.63 & 0.97 & 96.0\\
 Lasso - SLaT~\cite{cai2017three} & 39.2 & 0.60 & 1.15 & 99.6 && 37.8 & 0.70 & 0.99 & 98.3\\
 P-ULA & 38.9 & 0.45 & \textbf{1.82} & 98.7 && 37.1 & 0.57 & 1.59 & 94.9\\ 
 HMC     & 40.0 & \textbf{0.62} & 1.47 & \textbf{99.7}&& 36.4 & 0.64 & 1.59 & 98.5\\
 PP-ULA   & \textbf{40.3} & \textbf{0.62} & \underline{1.51} & \textbf{99.7}&& \textbf{38.6}& \textbf{0.71} &\textbf{1.64} & \textbf{98.7}\\
 \bottomrule\\
\end{tabular}
\caption{ PSNR, SSIM, CNR and segmentation overall accuracy (OA) for simulated data.}
\label{tab:psnr_ssim_cnr_simu1_simu2}
\end{table}

From Table~\ref{tab:res_simu1_simu2}, PP-ULA is 5.7 and 6.6 times faster than HMC on Simu1 and Simu2 and has better mixing properties, as shown by the MSJ per second.
Visual results from Fig.~\ref{fig:TRF_simu} and CNR values in Table~\ref{tab:psnr_ssim_cnr_simu1_simu2} show that the contrast obtained with PP-ULA is better than with competitors on Simu2, and is second best after P-ULA on Simu1. However, it should be noted that the PSNR and SSIM obtained on Simu1 with P-ULA are much lower than with the other methods.
In addition, the PSNR and SSIM values from Table~\ref{tab:psnr_ssim_cnr_simu1_simu2} obtained with PP-ULA are equivalent or higher than all competitors for these two experiments.
Visual segmentation results are shown in Fig.~\ref{fig:seg_simu}, and OA values can be found in Table~\ref{tab:psnr_ssim_cnr_simu1_simu2}. For these simulated images, more pixels are correctly labeled with PP-ULA than with competitors.

\begin{figure}[!htb]
\centering
\setlength\tabcolsep{0.01pt}
\begin{tabular}{ccccccc}
%% Simu1
\includegraphics[height=0.13\textwidth, width=0.14\textwidth]{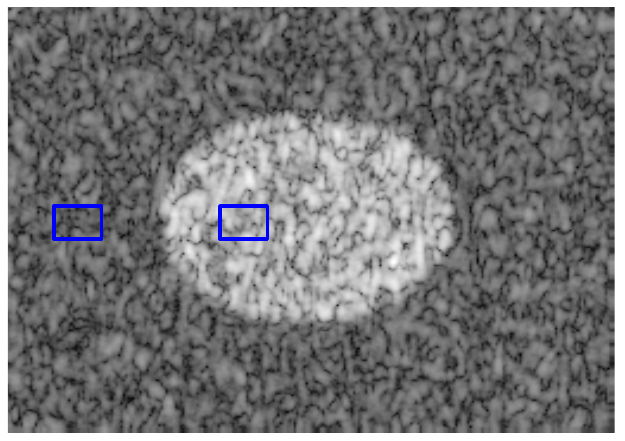}&
\includegraphics[height=0.13\textwidth, width=0.14\textwidth]{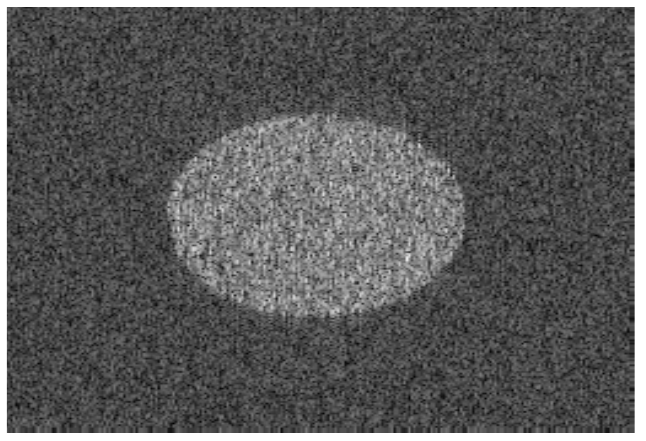}&
\includegraphics[height=0.13\textwidth, width=0.14\textwidth]{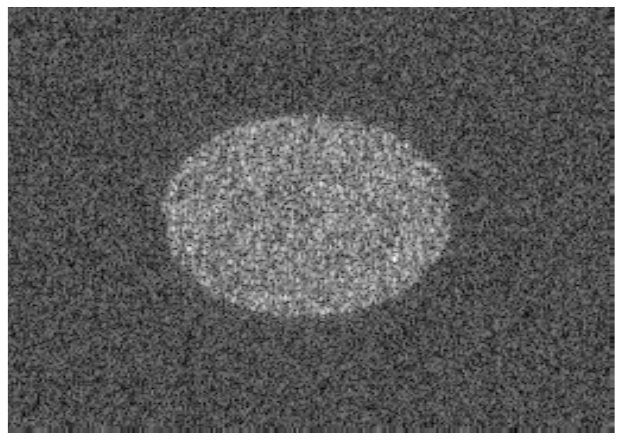}&
\includegraphics[height=0.13\textwidth, width=0.14\textwidth]{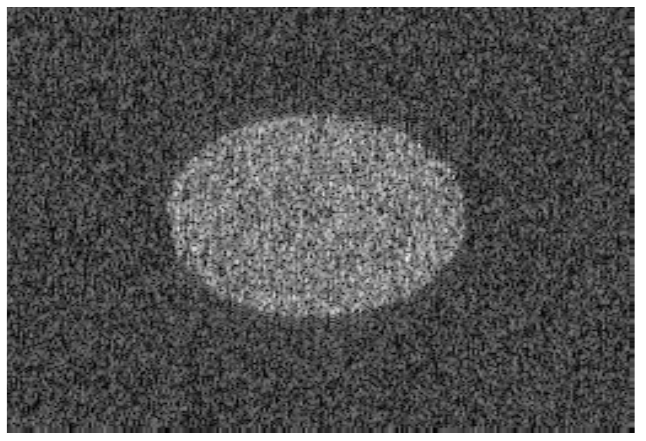}&
\includegraphics[height=0.13\textwidth, width=0.14\textwidth]{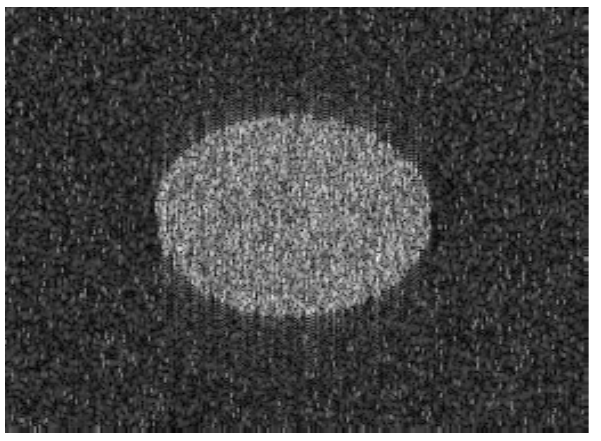}&
\includegraphics[height=0.13\textwidth, width=0.14\textwidth]{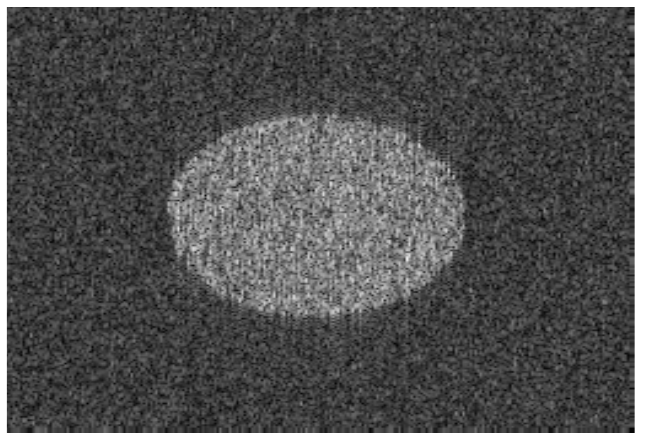}&
\includegraphics[height=0.13\textwidth, width=0.14\textwidth]{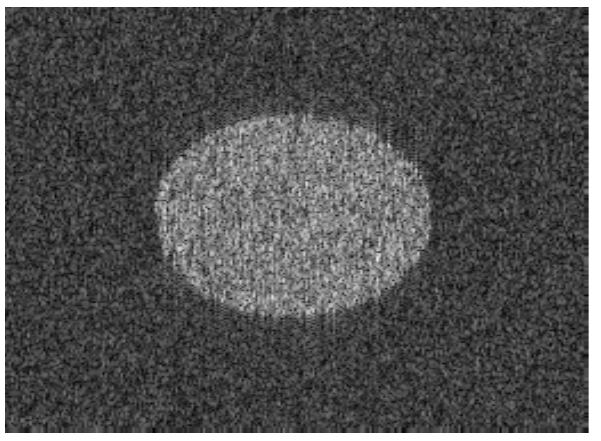}\\[-0.9ex]
%% Simu2
\includegraphics[height=0.13\textwidth, width=0.14\textwidth]{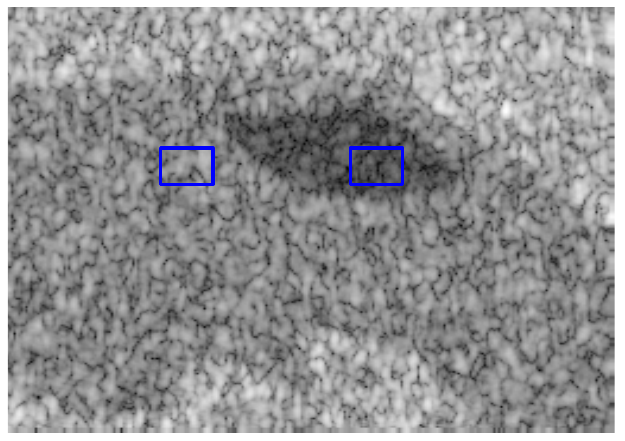}&
\includegraphics[height=0.13\textwidth, width=0.14\textwidth]{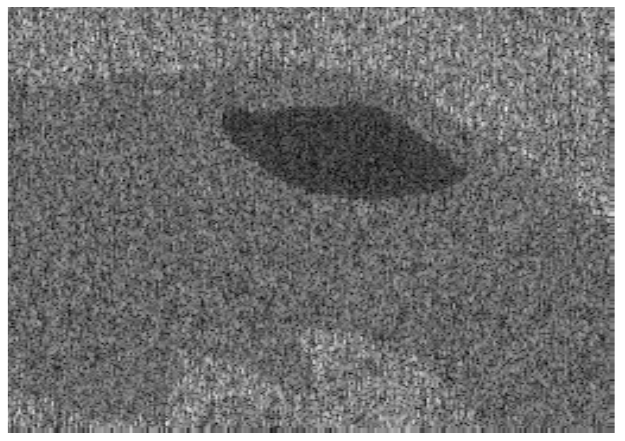}&
\includegraphics[height=0.13\textwidth, width=0.14\textwidth]{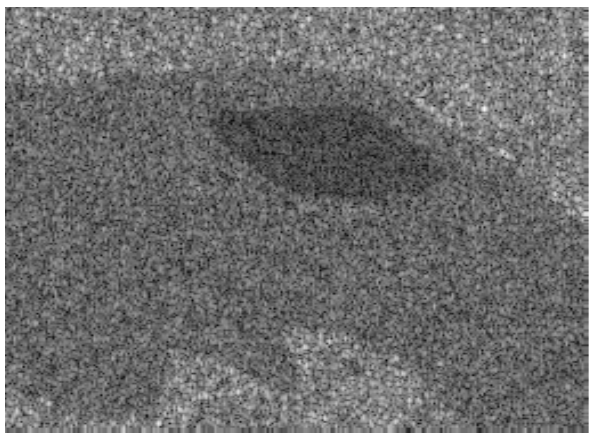}&
\includegraphics[height=0.13\textwidth, width=0.14\textwidth]{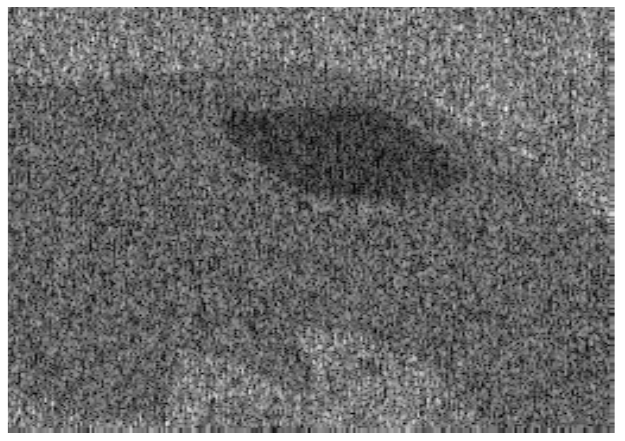}&
\includegraphics[height=0.13\textwidth, width=0.14\textwidth]{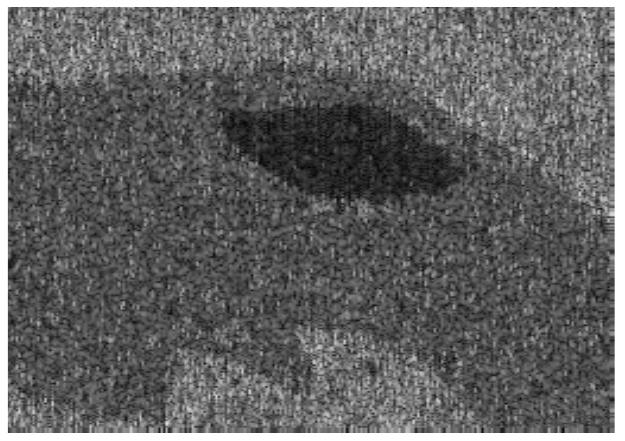}&
\includegraphics[height=0.13\textwidth, width=0.14\textwidth]{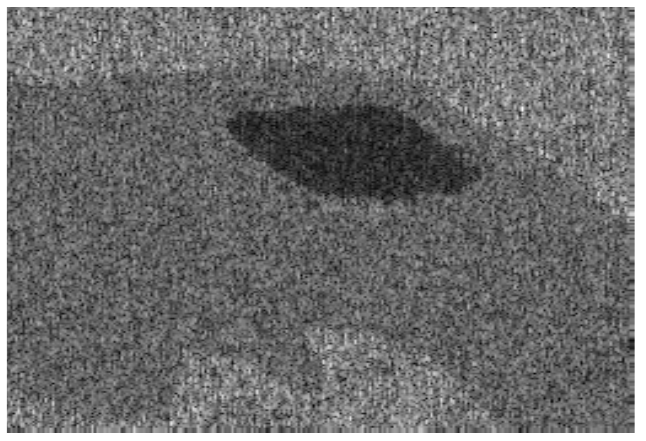}&
\includegraphics[height=0.13\textwidth, width=0.14\textwidth]{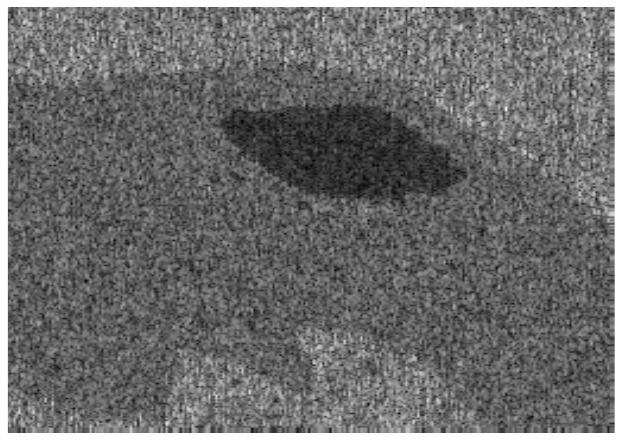}\\[-0.9ex]
\end{tabular}
\caption{B-mode visualization on simulated images. Top: Simu1. Bottom: Simu2. Left to right: RF image, TRF: ground-truth, Wiener, Lasso, P-ULA, HMC, PP-ULA. Blue boxes indicate regions used for the CNR.}
\label{fig:TRF_simu}
\end{figure}

\begin{figure}[!htb]
\centering
\setlength\tabcolsep{0.01pt}
\begin{tabular}{cccccc}
%% Simu1
\includegraphics[height=0.13\textwidth, width=0.14\textwidth]{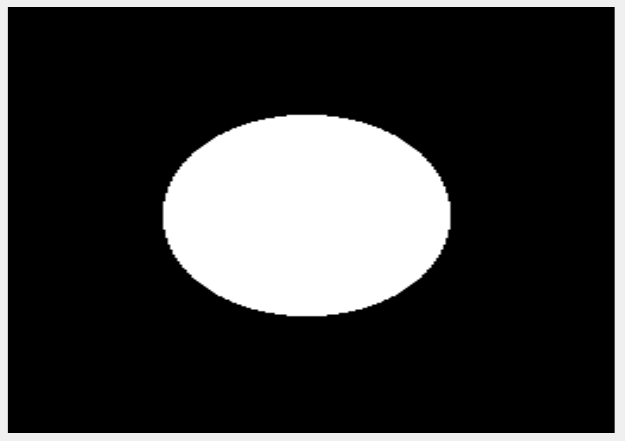}&
\includegraphics[height=0.13\textwidth, width=0.14\textwidth]{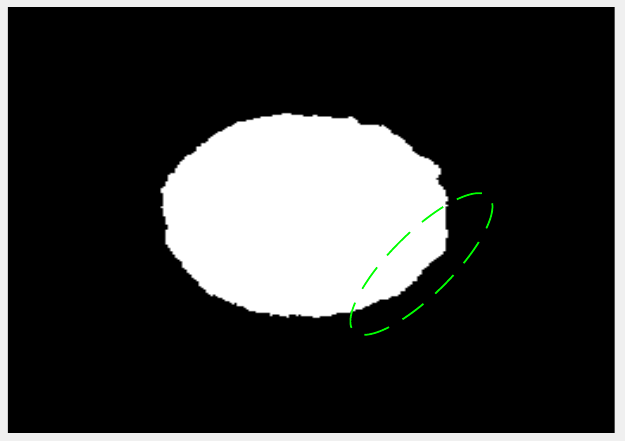}&
\includegraphics[height=0.13\textwidth, width=0.14\textwidth]{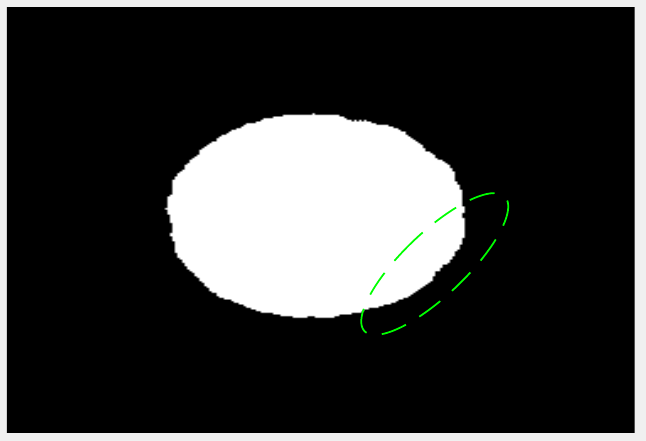}&
\includegraphics[height=0.13\textwidth, width=0.14\textwidth]{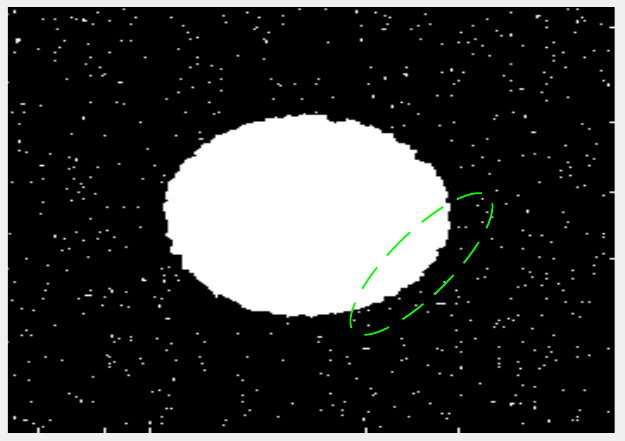}&
\includegraphics[height=0.13\textwidth, width=0.14\textwidth]{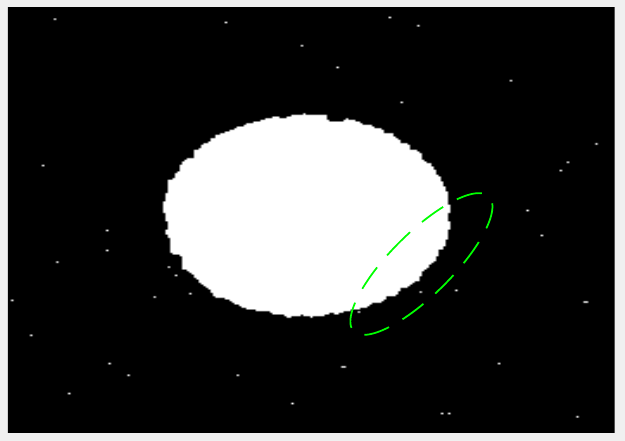}&
\includegraphics[height=0.13\textwidth, width=0.14\textwidth]{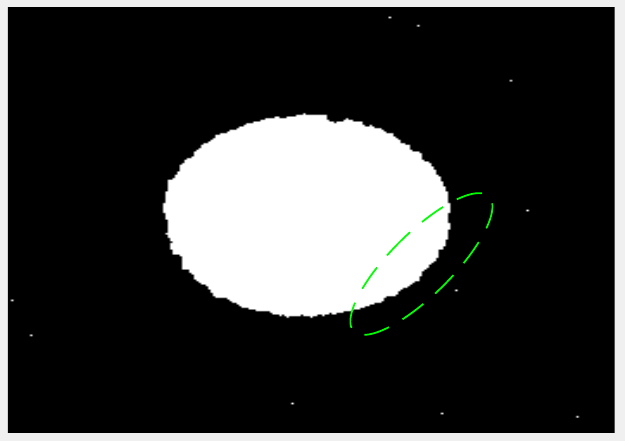}\\
%% Simu2
\includegraphics[height=0.13\textwidth, width=0.14\textwidth]{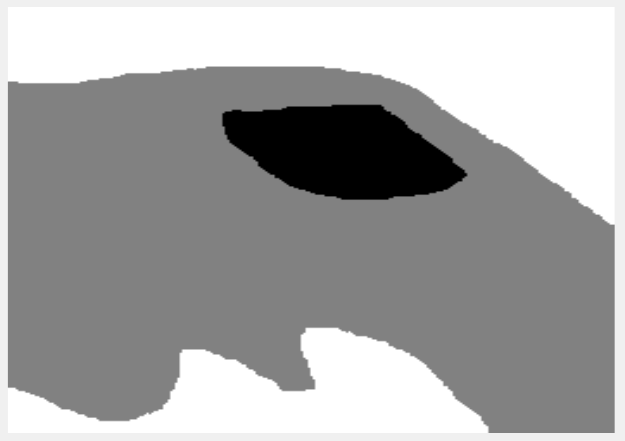}&
\includegraphics[height=0.13\textwidth, width=0.14\textwidth]{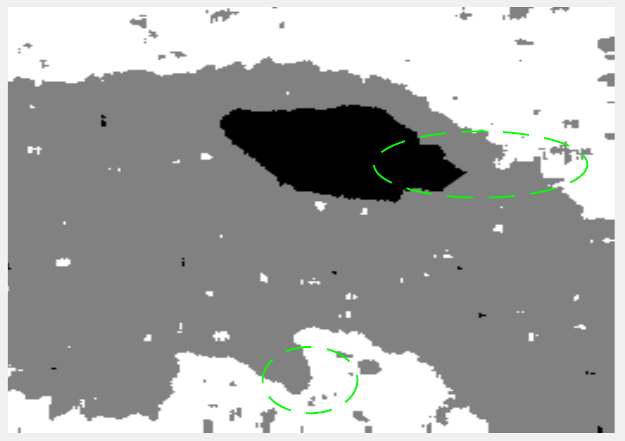}&
\includegraphics[height=0.13\textwidth, width=0.14\textwidth]{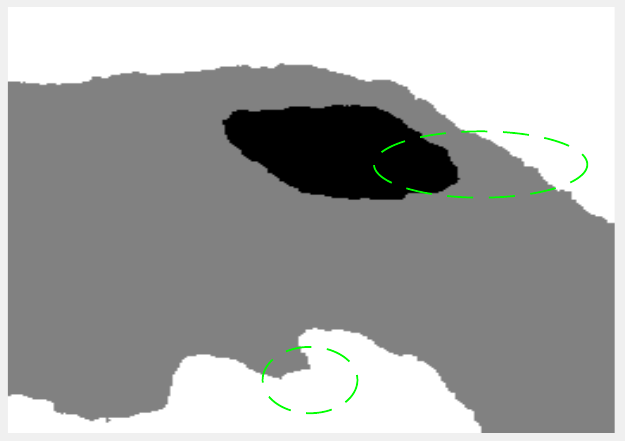}&
\includegraphics[height=0.13\textwidth, width=0.14\textwidth]{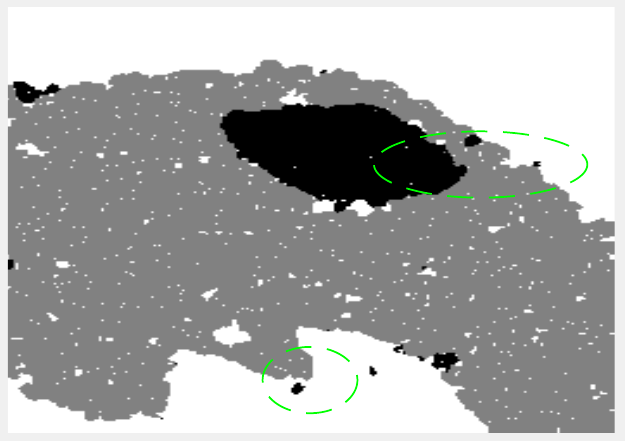}&
\includegraphics[height=0.13\textwidth, width=0.14\textwidth]{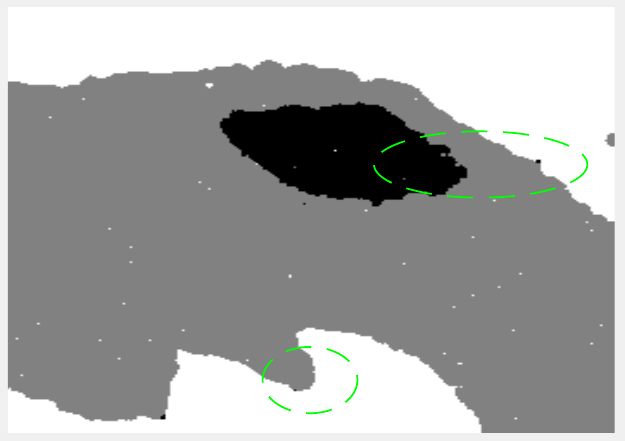}&
\includegraphics[height=0.13\textwidth, width=0.14\textwidth]{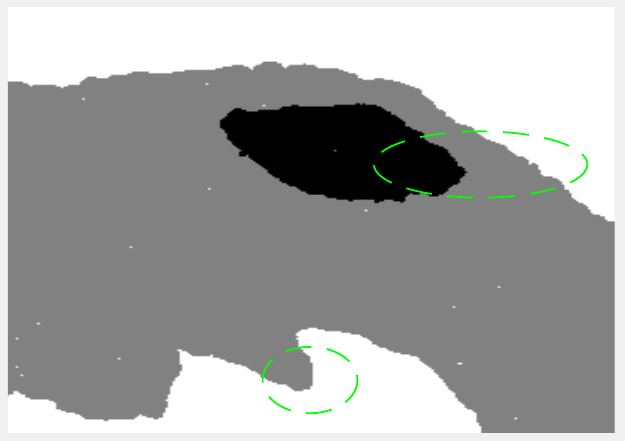}\\
\end{tabular}
\caption{Segmentation. Top: Simu1, bottom: Simu2. Left to right: ground-truth, Otsu, SLaT, P-ULA, HMC, PP-ULA. Main differences are circled in green.}
\label{fig:seg_simu}
\end{figure}

\begin{table}[!htb]
\centering
\setlength\tabcolsep{3pt}
\begin{tabular}{lcrcccccccc}
\toprule
 && && \multicolumn{2}{c}{Iterations} && \multicolumn{2}{c}{Time} && \multicolumn{1}{c}{Mixing property}\\
 \cmidrule{5-6} \cmidrule{8-9} \cmidrule{11-11}
 && && Burn-in & Total && Duration & PP-ULA speed gain && MSJ (per s)\\
 \midrule
 \multirow{2}{*}{Kidney}
 &&HMC   && 7000  & 14000   && 4~h~23~min & 6.3  && 167\\
 &&PP-ULA&& 7000  & 14000   && 42~min     & 1    && 657\\
 \midrule
  \multirow{2}{*}{Thyroid}
 &&HMC   && 3000  & 6000   && 2~h~09~min & 3.7  && 175\\
 &&PP-ULA&& 3000  & 6000   && 35~min     & 1    && 950\\
 \midrule
   \multirow{2}{*}{Bladder}
 &&HMC   && 5000  & 10000   && 2~h~45~min & 5.2  && 13\\
 &&PP-ULA&& 5000  & 10000   && 32~min     & 1    && 1396\\
 \midrule
   \multirow{2}{*}{KidneyReal}
 &&HMC   && 5000  & 10000   && 1~h~49~min & 5.8  && 11\\
 &&PP-ULA&& 5000  & 10000   && 19~min     & 1    && 1361\\
 \bottomrule\\
\end{tabular}
\caption{Number of iterations, computational time and MSJ per s for experiments on the tissue-mimicking phantom and on real data.}
\label{tab:res_others}
\end{table}

\subsection{Results on a tissue-mimicking phantom and on real data}

The convergence of Algorithm~\ref{algo:gibbs} is also empirically observed for the experiments on the tissue-mimicking phantom and on real data, i.e. Kidney, Thyroid, Bladder and KidneyReal. As mentioned in Table~\ref{tab:res_others}, the proposed method leads to a significant acceleration since it is between 3.7 and 6.3 times faster than HMC on these experiments.
Visual results from Fig.~\ref{fig:TRF_real} and CNR values in Table~\ref{tab:psnr_ssim_cnr} show that the contrast obtained with  PP-ULA is better than with competitors on all these test images. 
In addition, the PSNR and SSIM values from Table~\ref{tab:psnr_ssim_cnr} obtained with PP-ULA on the Kidney experiment are equivalent or higher than all competitors.
Although the ground-truth of the segmentation is not available for these experiments, one can see from the visual segmentation results shown in Fig.~\ref{fig:seg_real}, that the segmentation based on the Potts model (PP-ULA and HMC) gives more homogeneous areas than Otsu, and recovers more details than SLaT. 

\begin{figure}[!htb]
\centering
\setlength\tabcolsep{0.01pt}
\begin{tabular}{c}
\begin{tabular}{cccccc}
%% Kidney
\includegraphics[height=0.15\textwidth, width=0.165\textwidth]{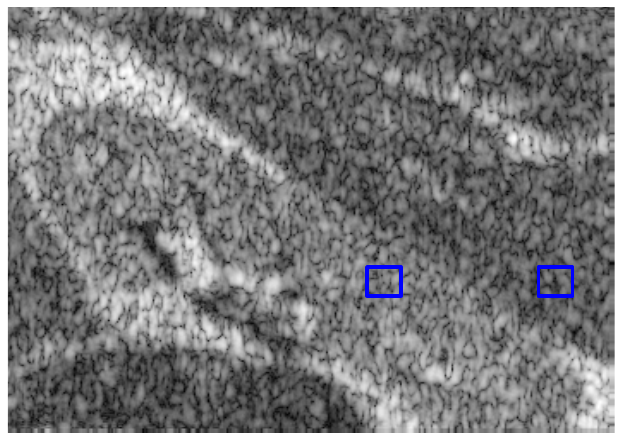}&
\includegraphics[height=0.15\textwidth, width=0.165\textwidth]{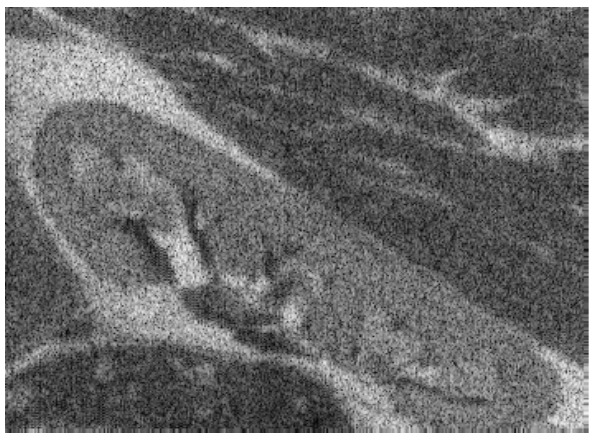}&
\includegraphics[height=0.15\textwidth, width=0.165\textwidth]{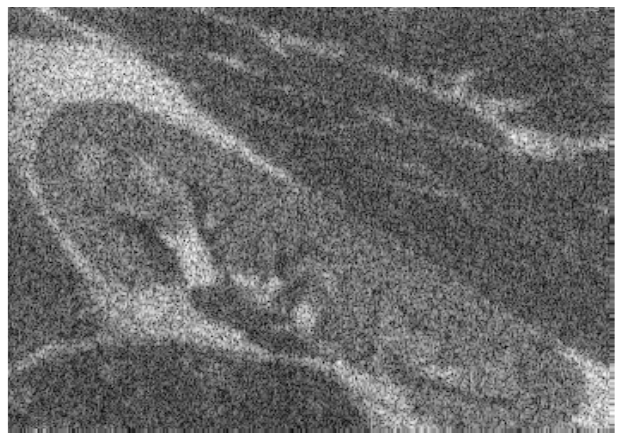}&
\includegraphics[height=0.15\textwidth, width=0.165\textwidth]{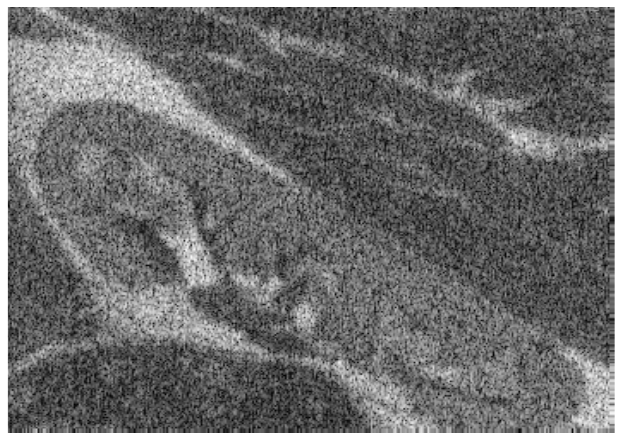}&
\includegraphics[height=0.15\textwidth, width=0.165\textwidth]{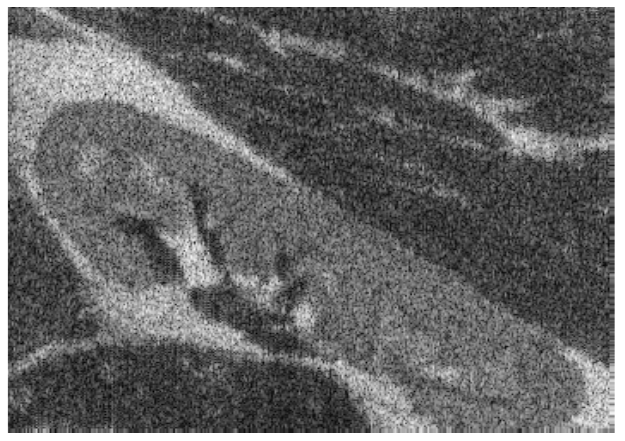}&
\includegraphics[height=0.15\textwidth, width=0.165\textwidth]{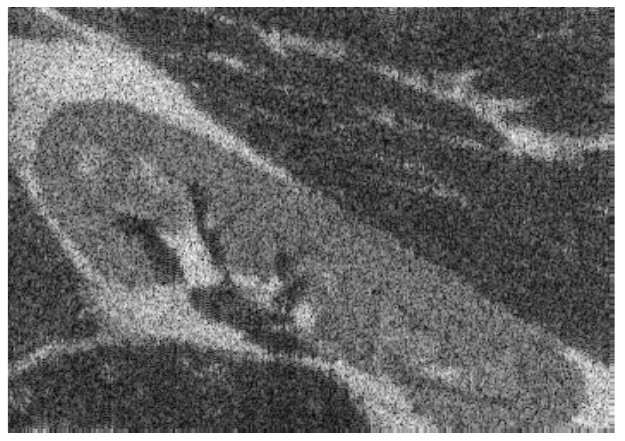}\\
\end{tabular}\\(a)\\
\begin{tabular}{ccccc}
%% Thyroid
\includegraphics[height=0.16\textwidth, width=0.2\textwidth]{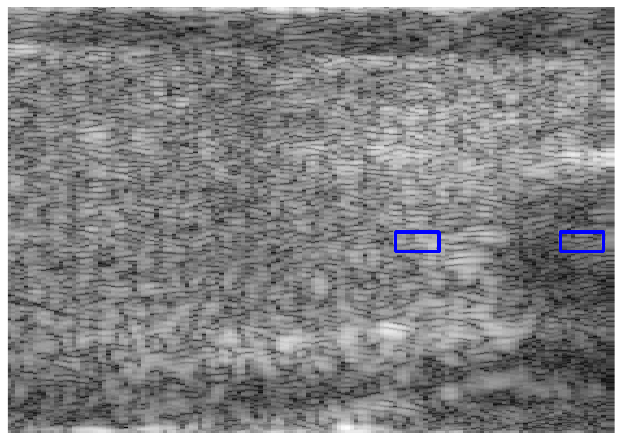}&
\includegraphics[height=0.16\textwidth, width=0.2\textwidth]{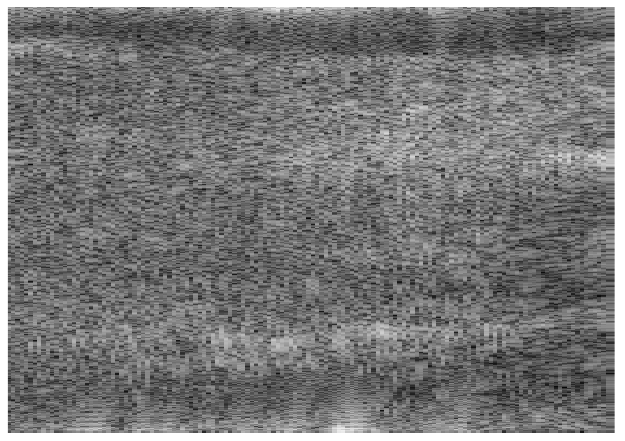}&
\includegraphics[height=0.16\textwidth, width=0.2\textwidth]{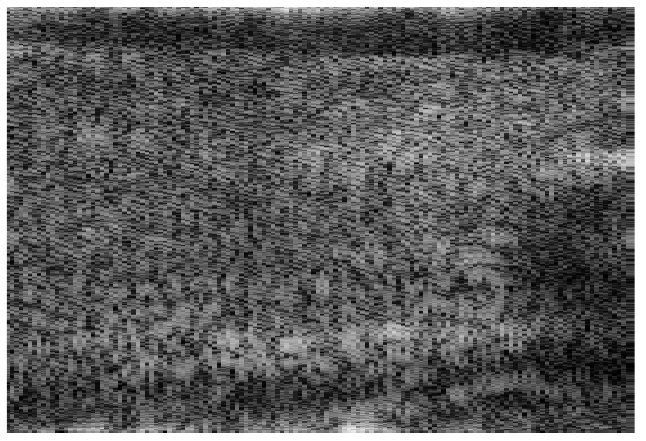}&
\includegraphics[height=0.16\textwidth, width=0.2\textwidth]{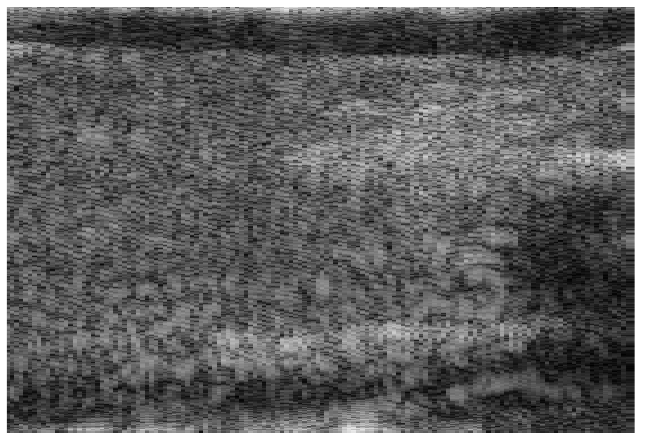}&
\includegraphics[height=0.16\textwidth, width=0.2\textwidth]{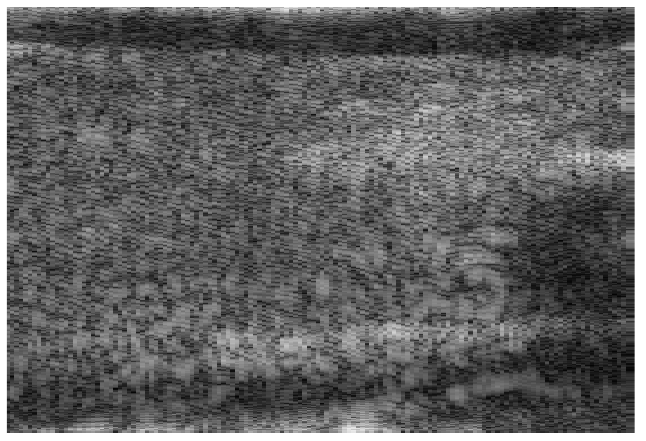}\\[-0.9ex]
%% Bladder
\includegraphics[height=0.16\textwidth, width=0.2\textwidth]{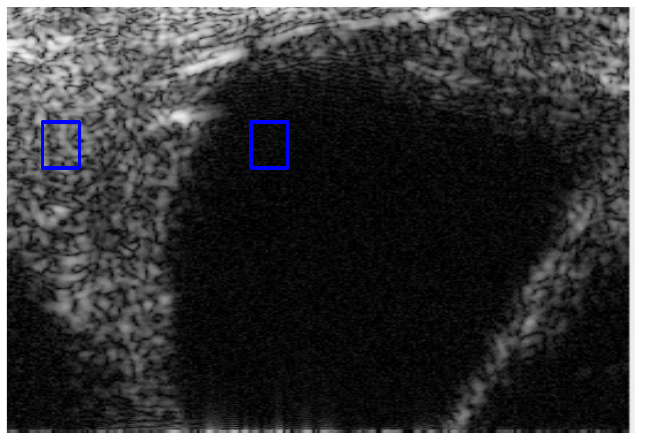}&
\includegraphics[height=0.16\textwidth, width=0.2\textwidth]{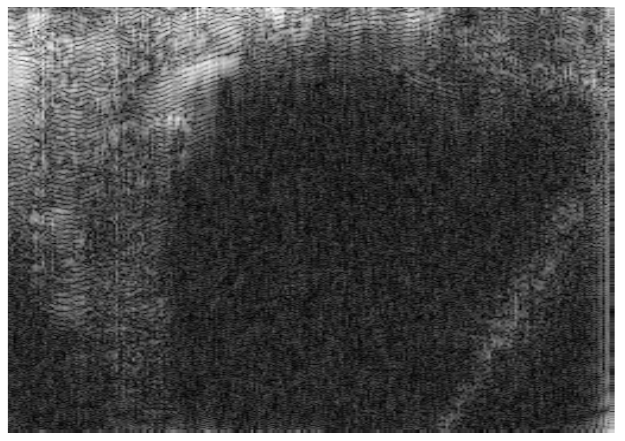}&
\includegraphics[height=0.16\textwidth, width=0.2\textwidth]{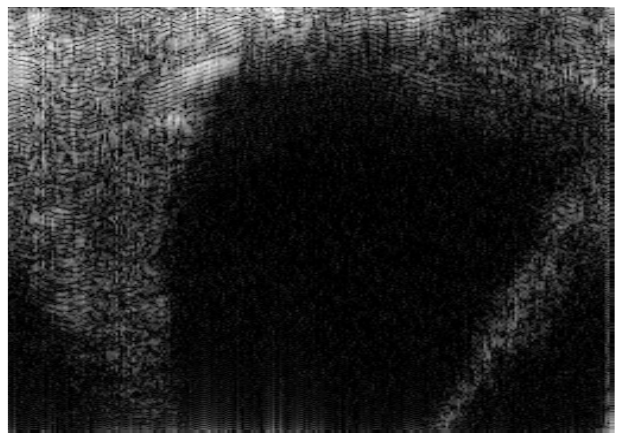}&
\includegraphics[height=0.16\textwidth, width=0.2\textwidth]{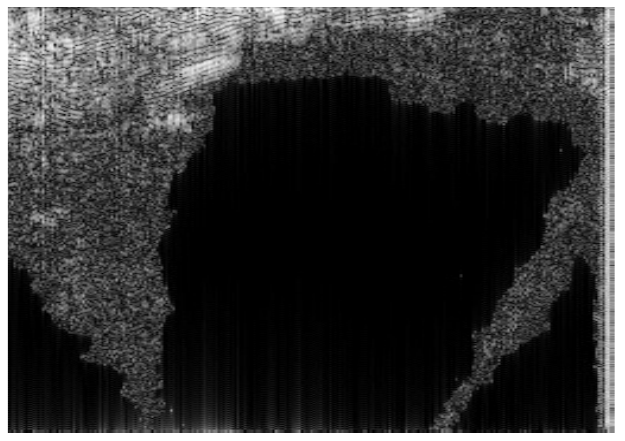}&
\includegraphics[height=0.16\textwidth, width=0.2\textwidth]{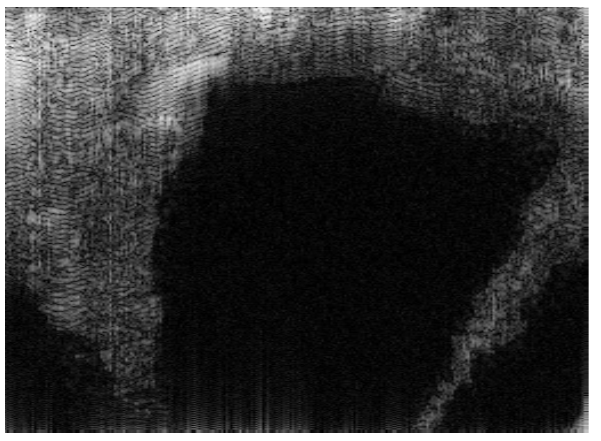}\\[-0.9ex]
%% KidneyReal
\includegraphics[height=0.16\textwidth, width=0.2\textwidth]{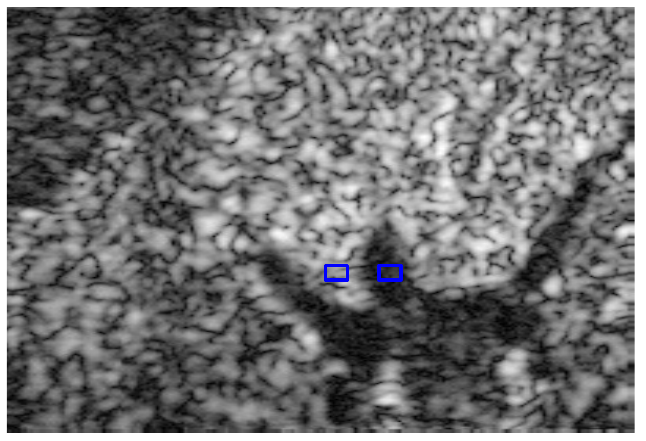}&
\includegraphics[height=0.16\textwidth, width=0.2\textwidth]{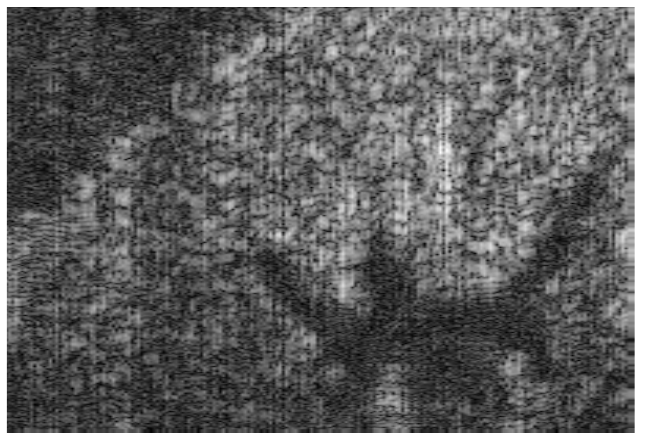}&
\includegraphics[height=0.16\textwidth, width=0.2\textwidth]{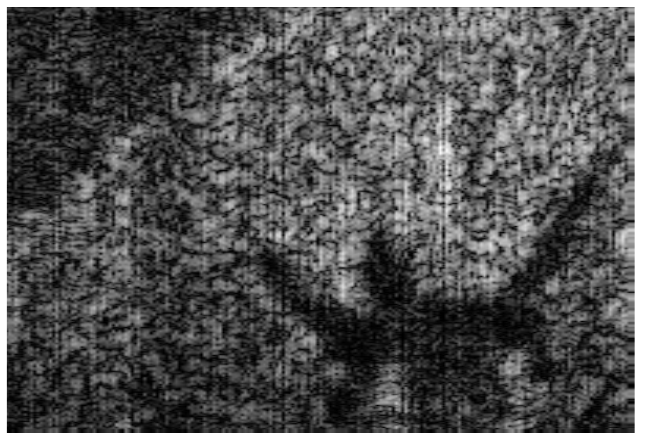}&
\includegraphics[height=0.16\textwidth, width=0.2\textwidth]{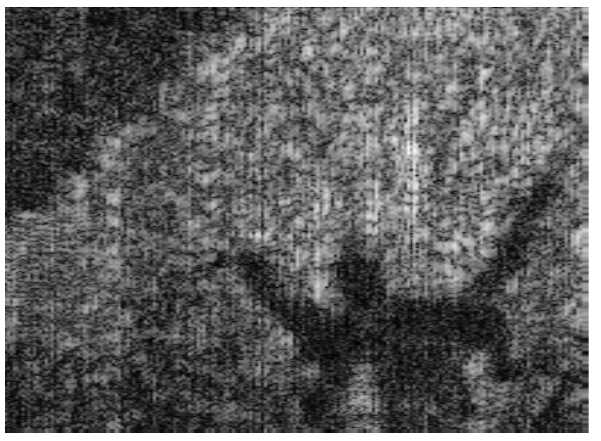}&
\includegraphics[height=0.16\textwidth, width=0.2\textwidth]{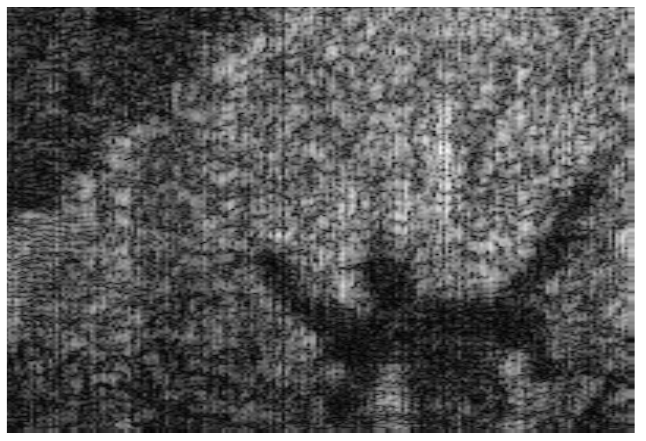}\\
\end{tabular}\\(b)\\
\end{tabular}
\caption{B-mode visualization. (a) Kidney; left to right: RF image, TRF: ground-truth, Wiener, Lasso, HMC, PP-ULA. (b)~Real \textit{in vivo} images; top to bottom: Thyroid, Bladder, KidneyReal; left to right: RF image, TRF: Wiener, Lasso, HMC, PP-ULA. Blue boxes indicate regions used for the CNR.}
\label{fig:TRF_real}
\end{figure}

\begin{figure}[!htb]
\centering
\setlength\tabcolsep{0.01pt}
\begin{tabular}{cccc}
%% Kidney
\includegraphics[height=0.16\textwidth, width=0.2\textwidth]{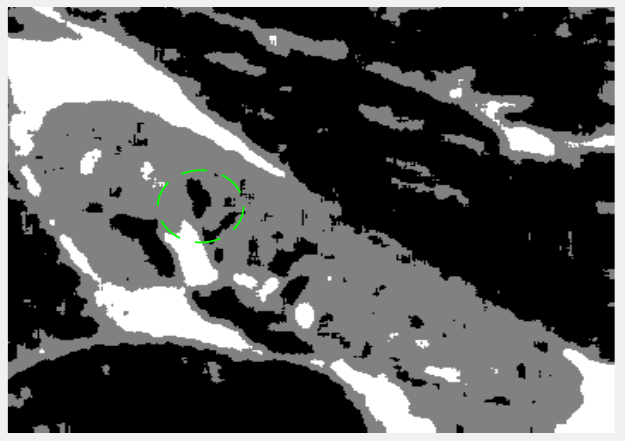}&
\includegraphics[height=0.16\textwidth, width=0.2\textwidth]{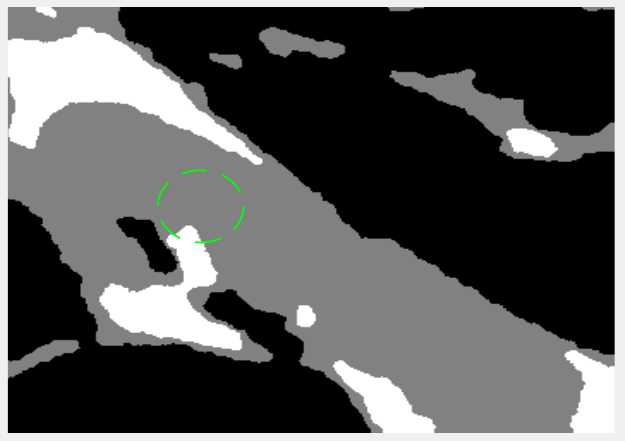}&
\includegraphics[height=0.16\textwidth, width=0.2\textwidth]{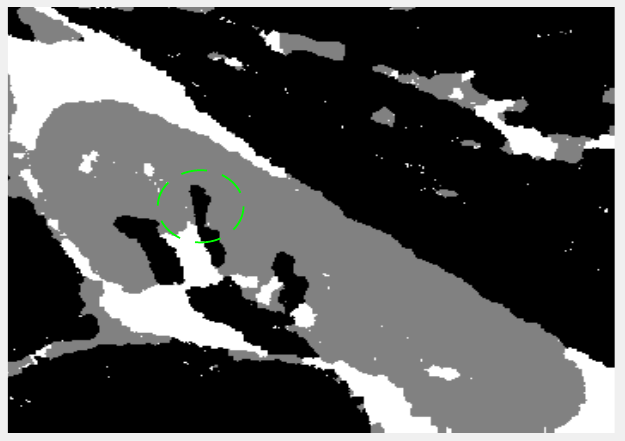}&
\includegraphics[height=0.16\textwidth, width=0.2\textwidth]{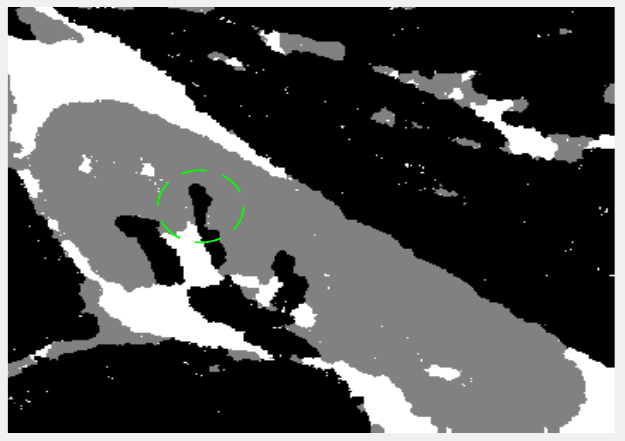}\\
%% Thyroid
\includegraphics[height=0.16\textwidth, width=0.2\textwidth]{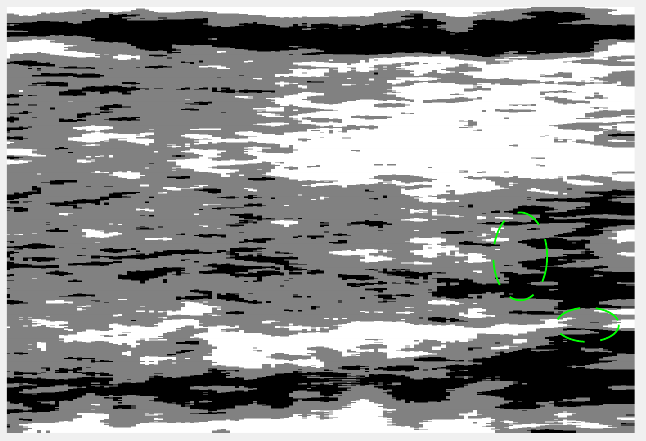}&
\includegraphics[height=0.16\textwidth, width=0.2\textwidth]{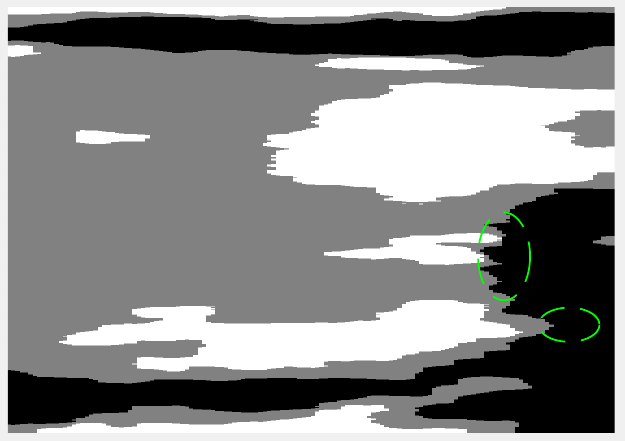}&
\includegraphics[height=0.16\textwidth, width=0.2\textwidth]{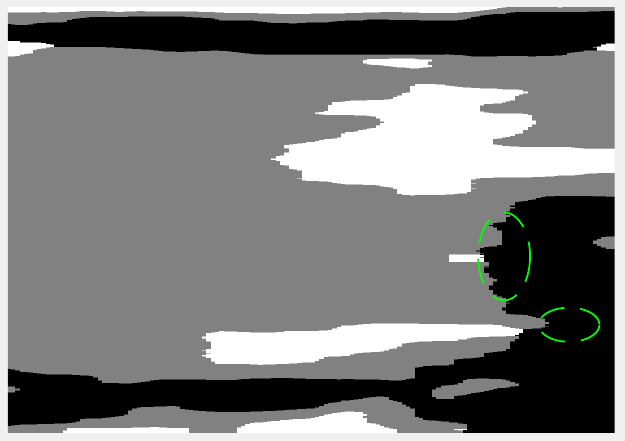}&
\includegraphics[height=0.16\textwidth, width=0.2\textwidth]{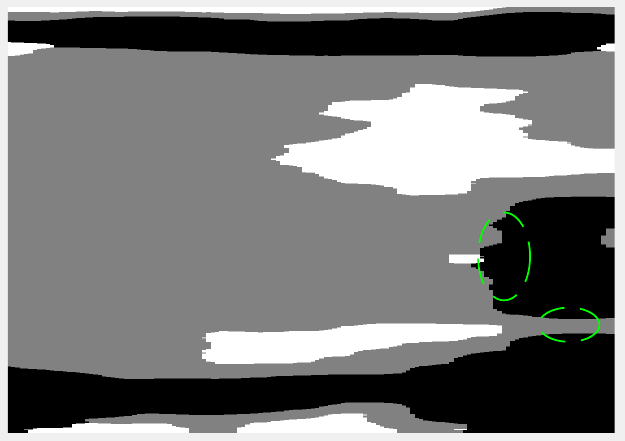}\\
%% Bladder
\includegraphics[height=0.16\textwidth, width=0.2\textwidth]{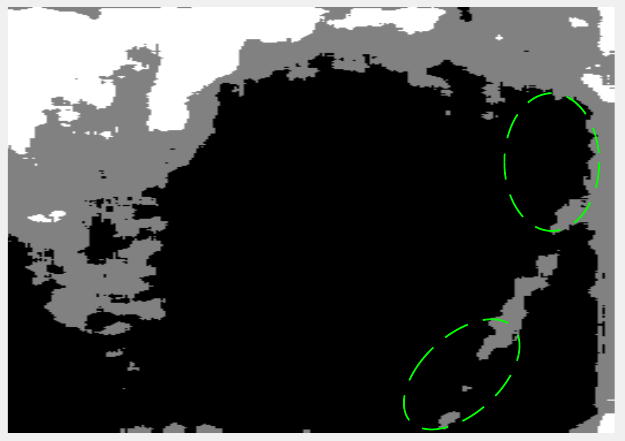}&
\includegraphics[height=0.16\textwidth, width=0.2\textwidth]{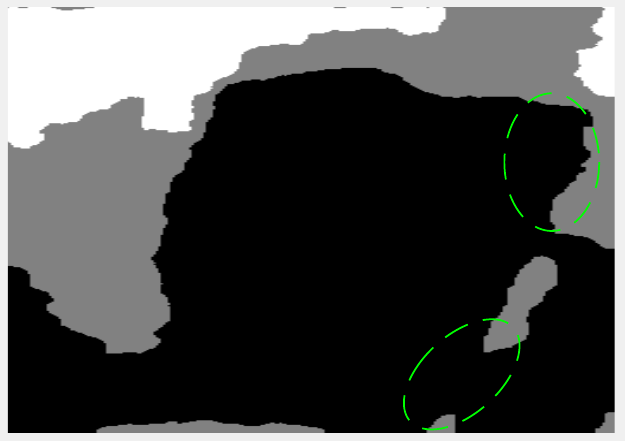}&
\includegraphics[height=0.16\textwidth, width=0.2\textwidth]{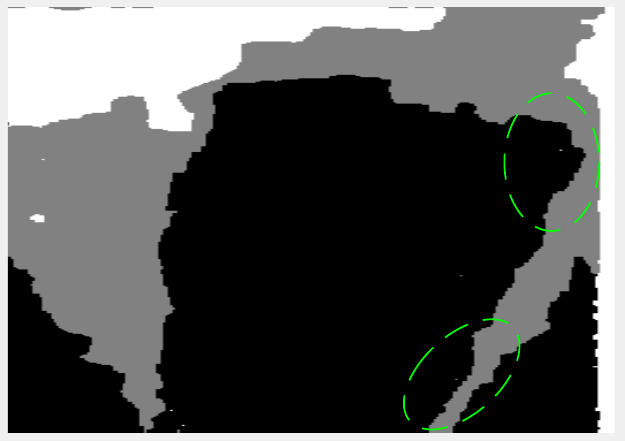}&
\includegraphics[height=0.16\textwidth, width=0.2\textwidth]{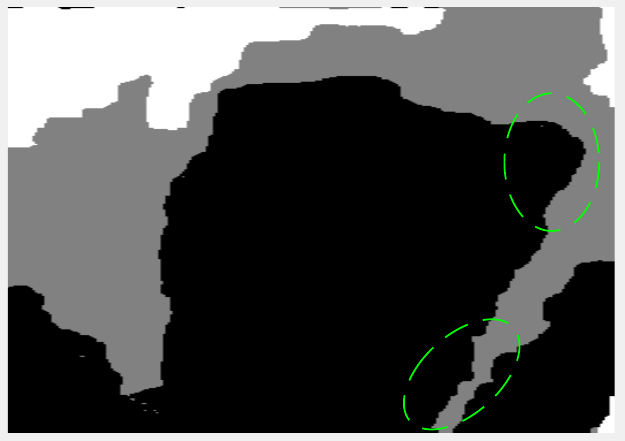}\\
%% KidneyReal
\includegraphics[height=0.16\textwidth, width=0.2\textwidth]{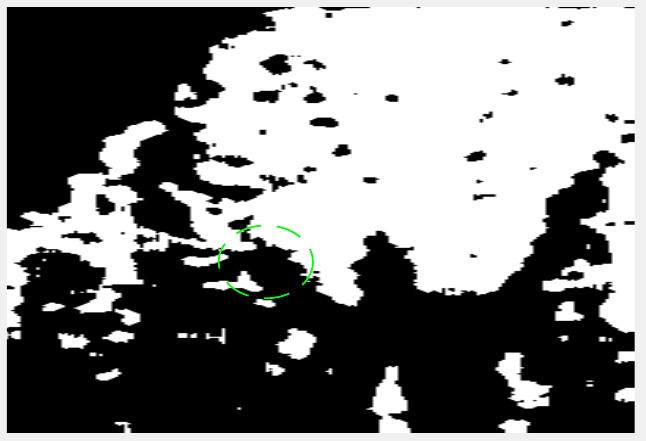}&
\includegraphics[height=0.16\textwidth, width=0.2\textwidth]{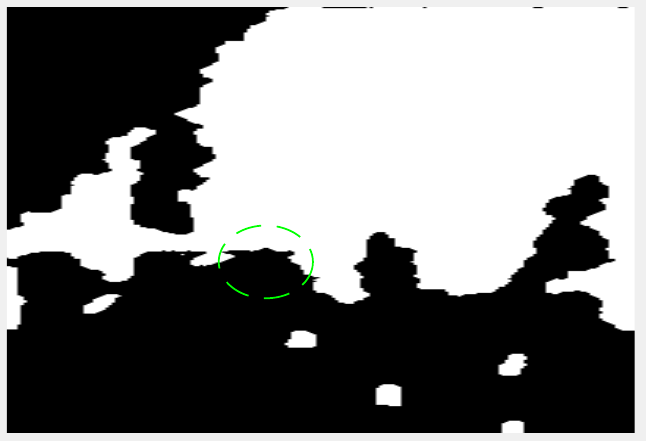}&
\includegraphics[height=0.16\textwidth, width=0.2\textwidth]{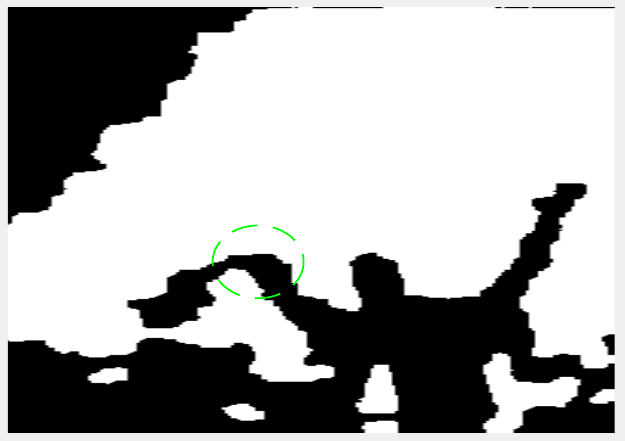}&
\includegraphics[height=0.16\textwidth, width=0.2\textwidth]{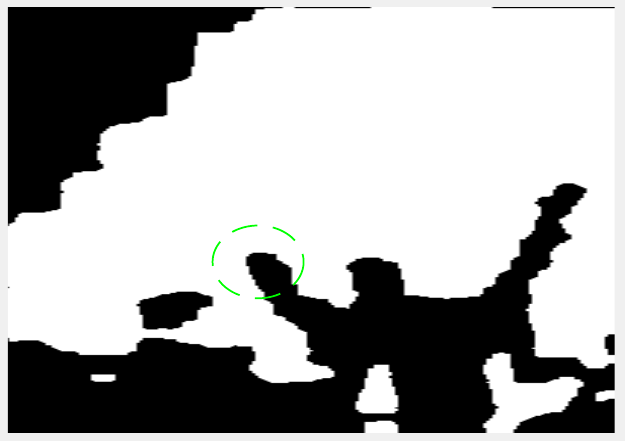}\\
\end{tabular}
\caption{Segmentation. Top to bottom: Kidney, Thyroid, Bladder, KidneyReal. Left to right: Otsu, SLaT, HMC, PP-ULA. Main differences are circled in green.}
\label{fig:seg_real}
\end{figure}

\begin{table}[!htb]
\centering
\setlength\tabcolsep{3pt}
\begin{tabular}{rcccccccccc}
\toprule
%%%% Simu1 and Simu2
 &&  \multicolumn{3}{c}{Kidney} && \multicolumn{1}{c}{Thyroid} && \multicolumn{1}{c}{Bladder} && \multicolumn{1}{c}{KidneyReal}\\
\cmidrule{3-5} \cmidrule{7-7} \cmidrule{9-9} \cmidrule{11-11}
         && PSNR & SSIM & CNR && CNR && CNR && CNR\\
     \midrule    
 Wiener && 27.6 & 0.58 & 0.66 && 0.56 && 1.66 && 1.61\\
 Lasso && 28.5 & 0.59 & 0.67 && 0.99 && 1.76 && 1.76\\
 HMC     && \textbf{29.5} & \textbf{0.62} & 1.10 && 1.52 && 2.23 && 1.88\\
 PP-ULA   &&  \underline{29.3} & \textbf{0.62} & \textbf{1.14} && \textbf{1.56} && \textbf{2.48} && \textbf{1.93}\\
 \bottomrule\\
\end{tabular}
\caption{ PSNR, SSIM and CNR results.}
\label{tab:psnr_ssim_cnr}
\end{table}

\section{Conclusion}
%The proposed sampling strategy for the TRF has proven to be much faster than HMC while ensuring satisfactory restoration, segmentation, and tissue characterization results for ultrasound images. 
We investigated a new method based on a preconditioned proximal unadjusted Langevin algorithm for the joint restoration and segmentation of ultrasound images, which showed faster convergence than an existing Hamiltonian Monte Carlo algorithm. Hence, the proposed method has the potential to speed-up the approach proposed in \cite{pereyra2012segmentation} for the segmentation of ultrasound images.
Another direction for future work is to extend this framework to a spatially variant, possibly unknown, PSF. 


\section*{Appendix}
\label{sec:annex}
In this section, after reminding results about the Langevin diffusion and its discretization using Euler's scheme, we provide details about the derivation of the proposed method \eqref{eq:p-ula_it} used to sample the TRF.  

%=============================================================================
\subsection{Discrete Langevin diffusion}
An $n$-dimensional Langevin diffusion is a continuous time Markov process $(x(t))_{t\in[0,+\infty[}$ taking its values in $\mathbb{R}^n$, which is the solution to the following stochastic differential equation \cite{roberts2002langevin},
\begin{equation}
(\forall t\in[0,+\infty[)~~\mathrm{d}x(t)=b(x(t))\mathrm{d}t+V(x(t))\mathrm{d}B(t),
\label{eq:langevin}
\end{equation}
where $(B(t))_{t\in[0,+\infty[}$ is a Brownian motion with values in $\mathbb{R}^n$, and for every $x\in\mathbb{R}^n$, $V(x)\in\mathbb{R}^{n\times n}$ is the volatility matrix and $b(x)=(b_i(x))_{1\leq i\leq n} \in \mathbb{R}^n$ is the drift term defined as
\begin{equation}
(\forall i \in\{1,\ldots,n\})~~~b_i(x)=\frac{1}{2}\sum_{j=1}^nA_{i,j}(x)\frac{\partial \log \pi(x)}{\partial x_j}+\mathrm{det}(A(x))^{\frac{1}{2}}\sum_{k=1}^n\frac{\partial}{\partial x_j}\left(A_{i,k}(x)\mathrm{det}(A(x))^{-\frac{1}{2}}\right),
\end{equation}
where $A(x)=V(x)V(x)^\top=(A_{i,j}(x))_{1\leq i,j\leq n}$ is a symmetric positive definite matrix, $\mathrm{det}(A(x))$ denotes its determinant, and $\pi$ is the density of the stationary distribution of the diffusion. Here, we take $(\forall x\in\mathbb{R}^n)$ $\pi(x)=p(x|y,\sigma^2,\alpha,\beta,z)$ defined in \eqref{eq:TRF_distri}. Euler's discretization scheme applied to \eqref{eq:langevin} leads to the following target posterior distribution, which can be used to generate a Langevin Markov chain.
\begin{equation}
(\forall t\in\mathbb{N})~~x^{(t+1)}=x^{(t)}+2\gamma b(x^{(t)})+\sqrt{2\gamma}A^{\frac{1}{2}}(x^{(t)})\omega^{(t)}.
\end{equation}
Hereabove, $\omega^{(t)}\sim\mathcal{N}(0,\mathbb{I}_n) $ and $\gamma>0$ is the discretization stepsize that controls the length of the jumps, while the scale matrix $A(\cdot)$ drives their direction. Instead of taking $A(\cdot)=\mathbb{I}_n$ as in the standard Metropolis adjusted Langevin algorithm, we follow \cite{marnissi2018majorize,stuart2004conditional} and use a preconditioning matrix $A$ to accelerate the Langevin scheme, which leads to 
\begin{equation}
x^{(t+1)}=x^{(t)}+\gamma A\nabla\log \pi(x)+\sqrt{2\gamma}A^{\frac{1}{2}}\omega^{(t+1)}.\label{eq:langevin_pi_A}
\end{equation}

%=============================================================================

\subsection{Approximation of the target diffusion}
For every $x\in\mathbb{R}^n$, let $f(x)=\|y-Hx\|^2/(2\sigma^2)$. From \eqref{eq:TRF_distri}, the target distribution statisfies the following relation,
\begin{equation}
(\forall x\in\mathbb{R}^n)~~\pi(x)=p(x|y,\sigma^2,\alpha,\beta,z)\propto \exp(-(f+g)(x)),
\end{equation}
where $(\forall x\in\mathbb{R}^n)$ $g(x)=\sum_{i=1}^n\beta_{z_i}^{-1}|x_i|^{\alpha_{z_i}}$.
Let $\gamma>0$ and $Q\in\mathcal{S}_n$. Following~\cite{pereyra2016proximal}, we replace $\pi$ by its Moreau approximation $\pi_\gamma^Q$ defined in \eqref{eq:approx_target} and recalled below, 
\begin{equation}
(\forall x\in\mathbb{R}^n)~~\pi_\gamma^Q(x)=\sup_{u\in\mathbb{R}^n}\pi(u)\exp\left(-\frac{\|u-x\|_{Q^{-1}}^2}{2\gamma}\right).
\end{equation}
Note that we dropped the normalization constant. Moreover, at the difference of~\cite{pereyra2016proximal}, we introduce the preconditioning matrix $Q$ for convergence acceleration purposes. When $Q$ is not specified, the identity matrix is used, \textit{i.e.} $Q=\mathbb{I}_n$. Hence, the approximated version of \eqref{eq:langevin_pi_A} reads
\begin{equation}
x^{(t+1)}=x^{(t)}+\gamma A\nabla\log \pi_\gamma^Q(x)+\sqrt{2\gamma}A^{\frac{1}{2}}\omega^{(t+1)}.\label{eq:langevin_pi_A_Q}
\end{equation}
We can then deduce the following result when $g$ is convex.

\begin{proposition} 
For every $\gamma>0$, $Q\in\mathcal{S}_n$ and $x\in\mathbb{R}^n$, if $(\forall k\in\{1,\ldots,K\})$ $\alpha_{k}\geq 1$, then we have
\begin{equation}
\nabla \log\pi_\gamma^Q(x)=Q^{-1}\frac{\mathrm{prox}_{\gamma (f+g)}^{Q}(x)-x}{\gamma}.
\end{equation}
\label{prop:nalbapi}
\end{proposition}
\begin{proof}
By definition of $\pi^{Q}_\gamma$, we have
\begin{equation}
(\forall x\in\mathbb{R}^n)~~\log\pi^{Q}_\gamma(x)=-{\color{black} (f+g)\left(\mathrm{prox}_{\gamma(f+g)}^{Q}(x)\right)-\frac{1}{2\gamma}\left\|\mathrm{prox}_{\gamma(f+g)}^{Q}(x)-x\right\|^2_{Q^{-1}}}.
\end{equation}
Hence, applying \cite[Lemma~2.5]{combettes2005signal} in the metric induced by $Q^{ -1}$ directly leads to the result.
\end{proof}

From Proposition~\ref{prop:nalbapi}, \eqref{eq:langevin_pi_A_Q} becomes
\begin{equation}
x^{(t+1)}=x^{(t)}+\gamma AQ^{-1}\frac{\mathrm{prox}_{\gamma (f+g)}^Q(x^{(t)})-x^{(t)}}{\gamma}+\sqrt{2\gamma}A^{\frac{1}{2}}\omega^{(t+1)}.
\label{eq:langevin_pi_A_Q_prox}
\end{equation}
It can be noted that, in Proposition~\ref{prop:nalbapi}, $g$ is assumed to be convex, which is not necessarily satified in our case. However, for simplicity, we take the discrete scheme \eqref{eq:langevin_pi_A_Q_prox} even in the nonconvex case. Finally, we take $A=Q$, which leads to
\begin{equation}
x^{(t+1)}=\mathrm{prox}_{\gamma (f+g)}^Q(x^{(t)})+\sqrt{2\gamma}Q^{\frac{1}{2}}\omega^{(t+1)}.
\label{eq:langevin_pi_Q}
\end{equation}

%=============================================================================

\subsection{Forward-backward approximation}
By definition, $f$ is differentiable on $\mathbb{R}^n$ and its gradient $\nabla f=H^\top(H\cdot-y)/\sigma^2$ is Lipschitz-continuous on $\mathbb{R}^n$. It is worth noting that the computation of the proximity operator of the sum of two functions is generally intractable \cite{pustelnik2017proximity}. Hence, as suggested in \cite{pereyra2016proximal}, we use a first-order Taylor expansion to approximate the proximity operator of $f+g$ and introduce a forward step in PP-ULA iteration. Let $o$ denotes Landau's notation.\footnote{Following Landau’s notation, we will write that $F(u) =o(\|u-x\|)$, where $F:\mathbb{R}^n\rightarrow \mathbb{R}$ and $x\in\mathbb{R}^n$, if $F(u)/\|u-x\|\rightarrow 0$ as $u\rightarrow x$.} Let $x\in\mathbb{R}^n$, using $(\forall u\in \mathbb{R}^n)$ $f(u)=f(x)+(u-x)^\top\nabla f(x)+o(\|u-x\|)$, we have
\begin{equation}
(f+g)(u) + \frac{1}{2\gamma}\|u-x\|^2_{Q^{-1}}=f(x)+ g(u)+ \frac{1}{2\gamma}\|u-x\|^2_{Q^{-1}} +(u-x)^\top \nabla f(x)+ o(\|u-x\|),
  \label{eq:Taylor_1}
\end{equation}
which can be re-written as
\begin{equation}
(f+g)(u) + \frac{1}{2\gamma}\|u-x\|^2_{Q^{-1}}=f(x)+g(u)+\frac{1}{2\gamma}\|u-x+\gamma Q \nabla f (x)\|^2_{Q^{-1}} -\frac{\gamma}{2}\|Q^{\frac{1}{2}}\nabla f(x)\|^2 + o(\|u-x\|).
\end{equation}
Hence, the proximity operator of $f+g$ can be expressed as follows,
\begin{align}
\mathrm{prox}_{\gamma (f+g)}^Q(x)&=\argmin_{u\in\mathbb{R}^n} \left((f+g)(u) + \frac{1}{2\gamma}\|u-x\|^2_{Q^{-1}}\right)\\
&=\argmin_{u\in\mathbb{R}^n} \left(g(u)+\frac{1}{2\gamma}\|u-x+\gamma Q \nabla f (x)\|^2_{Q^{-1}} + o(\|u-x\|)\right).
\end{align}
In addition, we have 
\begin{equation}
\mathrm{prox}_{\gamma g}^Q(x-\gamma Q \nabla f (x))=\argmin_{u\in\mathbb{R}^n} \left( g(u)+\frac{1}{2\gamma}\|u-x+\gamma Q \nabla f (x)\|^2_{Q^{-1}}\right).
\end{equation}
Therefore, when $\gamma$ is small, $\mathrm{prox}_{\gamma g}^Q(x-\gamma Q \nabla f (x))$ is a good approximation of $\mathrm{prox}_{\gamma (f+g)}^Q(x)$. Plugging this preconditioned forward-backward scheme \cite{combettes2011proximal} in \eqref{eq:langevin_pi_Q} leads to the proposed sampling method
\begin{equation}
x^{(t+1)}=\mathrm{prox}_{{\color{black}\gamma} g}^Q(x^{(t)}-{\color{black}\gamma} Q \nabla f (x^{(t)}))+\sqrt{2\gamma}Q^{\frac{1}{2}}\omega^{(t+1)}.
\end{equation}

\IEEEtriggeratref{37}
\bibliographystyle{IEEEbib}
\bibliography{refs}

\end{document}